 \pdfoutput=1
\newif\ifAnonimous\Anonimoustrue
\Anonimousfalse

\ifAnonimous
\documentclass[anonymous,letter,UKenglish,cleveref,autoref,thm-restate]{lipics-v2021}
\else
\documentclass[letter,UKenglish,cleveref,autoref,thm-restate]{lipics-v2021}
\fi

\nolinenumbers

\usepackage{microtype}
\usepackage{nccmath}
\usepackage[nocompress]{cite}
\usepackage{tikz}
\usetikzlibrary{automata, positioning}
\usepackage{tabularx}
\usepackage{booktabs}
\usepackage{xspace}
\usepackage[noEnd=true,indLines=false]{algpseudocodex}

\usepackage[only,llbracket,rrbracket]{stmaryrd}
\usepackage{mathrsfs}
\usepackage{mathtools}
\usepackage{adjustbox}

\renewcommand{\geq}{\geqslant}
\renewcommand{\leq}{\leqslant}
\renewcommand{\ge}{\geq}
\renewcommand{\le}{\leq}

\newcommand{\ie}{i.e.,~}
\newcommand{\eg}{e.g.~}

\newcommand{\sse}{\subseteq}

\newcommand{\abs}[1] {\ensuremath\left|#1\right|}
\newcommand{\ceil}[1] {\ensuremath\left\lceil#1\right\rceil}

\newcommand{\bigO}{\mathcal{O}}

\newcommand{\gen}[1]{\langle #1 \rangle}

\newcommand{\dom}{\operatorname{dom}}
\newcommand{\ran}{\operatorname{ran}}

\newcommand{\Ptime}{\ensuremath{\mathsf{P}}\xspace}
\newcommand{\NPOLYLOGTIME}{\ensuremath{\mathsf{NPOLYLOGTIME}}\xspace}

\newcommand{\ACz}{\ensuremath{\mathsf{AC}^0}\xspace}
\newcommand{\qACz}{\ensuremath{\mathsf{qAC}^0}\xspace}
\newcommand{\NC}{\ensuremath{\mathsf{NC}}\xspace}

\newcommand{\LOGSPACE}{\ensuremath{\mathsf{L}}\xspace}
\newcommand{\NL}{\ensuremath{\mathsf{NL}}\xspace}
\newcommand{\SL}{\ensuremath{\mathsf{SL}}\xspace}
\newcommand{\PTIME}{\ensuremath{\mathsf{P}}\xspace}
\newcommand{\NP}{\ensuremath{\mathsf{NP}}\xspace}
\newcommand{\GI}{\ensuremath{\mathsf{GI}}}
\newcommand{\PSPACE}{\ensuremath{\mathsf{PSPACE}}\xspace}

\renewcommand{\L}{\LOGSPACE}

\newcommand{\vU}{\ensuremath{\mathbf{U}}}
\newcommand{\vV}{\ensuremath{\mathbf{V}}}
\newcommand{\vT}{\ensuremath{\mathbf{T}}}
\newcommand{\vH}{\ensuremath{\mathbf{H}}}
\newcommand{\vAb}{\ensuremath{\mathbf{Ab}}}

\newcommand{\vGSol}{\ensuremath{\mathbf{G_{\text{sol}}}}}
\newcommand{\vG}{\ensuremath{\mathbf{G}}}
\newcommand{\vSl}{\ensuremath{\mathbf{Sl}}}
\newcommand{\vCl}{\ensuremath{\mathbf{Cl}}}
\newcommand{\vSI}{\ensuremath{\mathbf{SIS}}}
\newcommand{\vCom}{\ensuremath{\mathbf{Com}}}

\newcommand{\vBS}{\ensuremath{\mathbf{BS}}}  % Brandt semigroup
\newcommand{\vBM}{\ensuremath{\mathbf{BM}}} % Brandt monoid

\newcommand{\vId}[1]{\ensuremath{\llbracket #1 \rrbracket}}

\newcommand{\ISym}{\cI}
\newcommand{\Sym}{\mathrm{Perm}}

\newcommand*{\gH}[1][]{\mathrel{\mathcal{H}_{#1}}}
\newcommand*{\gL}[1][]{\mathrel{\mathcal{L}_{#1}}}
\newcommand*{\gR}[1][]{\mathrel{\mathcal{R}_{#1}}}
\newcommand*{\gD}[1][]{\mathrel{\mathcal{D}_{#1}}}
\newcommand*{\gJ}[1][]{\mathrel{\mathcal{J}_{#1}}}
\newcommand*{\gX}[1][]{\mathrel{\mathcal{X}_{#1}}}
\newcommand*{\gLle}[1][]{\leq_{\mathcal{L}_{#1}}}

\newcommand*{\gLge}[1][]{\geq_{\mathcal{L}_{#1}}}

\newcommand*{\gRle}[1][]{\leq_{\mathcal{R}_{#1}}}
\newcommand*{\gRlt}[1][]{<_{\mathcal{R}_{#1}}}
\newcommand*{\gRge}[1][]{\geq_{\mathcal{R}_{#1}}}
\newcommand*{\gRgt}[1][]{>_{\mathcal{R}_{#1}}}
\newcommand*{\gJle}[1][]{\leq_{\mathcal{J}_{#1}}}

\newcommand*{\gJge}[1][]{\geq_{\mathcal{J}_{#1}}}

\newcommand*{\gXle}[1][]{\leq_{\mathcal{X}_{#1}}}

\newcommand*{\gXge}[1][]{\geq_{\mathcal{X}_{#1}}}

\newcommand\nindent{.5pt}
\newcommand\noverline[1]{%
	\kern\nindent\overline{\kern-\nindent#1\kern-\nindent}\kern\nindent}

\newcommand{\ov}[1]{\noverline{#1}}

\newcommand{\cA}{\mathcal{A}}

\newcommand{\cL}{\mathcal{L}}

\newcommand{\cC}{\mathcal{C}}
\newcommand{\cD}{\mathcal{D}}

\newcommand{\cI}{\mathcal{I}}

\newcommand{\cX}{\mathcal{X}}

\theoremstyle{plain}
\newtheorem{question}[theorem]{Question}
\newtheorem{problem}[theorem]{Open Problem}
\makeatletter % bold heading for remarks
\def\th@remark{%
  \thm@headfont{%
    \textcolor{lipicsGray}{$\blacktriangleright$}\nobreakspace\sffamily\bfseries}%
  \normalfont % body font
  \thm@preskip\topsep \divide\thm@preskip\tw@
  \thm@postskip\thm@preskip
}
\makeatother

\theoremstyle{plain}
\newtheorem{maintheorem}{Theorem}

\newtheorem{maincorollary}[maintheorem]{Corollary}

\makeatletter
\providecommand\iitem{}
\providecommand\qitem{}
\newcommand\decproblem@iitem@label{\rlap{Input.}\phantom{Question.}}
\newcommand\decproblem@qitem@label{Question.}
\newenvironment{decproblem}{%
  \begin{description}\begin{samepage}%
  \renewcommand{\iitem}{\item[\decproblem@iitem@label]}%
  \renewcommand{\qitem}{\item[\decproblem@qitem@label]}%
}{%
  \end{samepage}\end{description}%
}
\makeatother

\newcommand{\dMemb}[2][]{\textup{\textsc{memb${}_{\mathbf{#1}}\expandafter\ifx\expandafter\relax\detokenize{#2}\relax\else(#2)\fi$}}}
\newcommand{\dConj}[2][]{\textup{\textsc{conj${}_{\mathbf{#1}}\expandafter\ifx\expandafter\relax\detokenize{#2}\relax\else(#2)\fi$}}}
\newcommand{\dMembS}[2][]{\textup{\textsc{memb${}^{\sharp}_{\mathbf{#1}}\expandafter\ifx\expandafter\relax\detokenize{#2}\relax\else(#2)\fi$}}}
\newcommand{\dConjS}[2][]{\textup{\textsc{conj${}^{\sharp}_{\mathbf{#1}}\expandafter\ifx\expandafter\relax\detokenize{#2}\relax\else(#2)\fi$}}}
\newcommand{\dMGS}[2][]{\textup{\textsc{mgs${}_{\mathbf{#1}}\expandafter\ifx\expandafter\relax\detokenize{#2}\relax\else(#2)\fi$}}}
\newcommand{\dRequiv}[2][]{\textup{\textsc{\ensuremath{\gR}-equiv${}_{\mathbf{#1}}\expandafter\ifx\expandafter\relax\detokenize{#2}\relax\else(#2)\fi$}}}
\newcommand{\dEqn}[2][]{\textup{\textsc{eqn${}_{\mathbf{#1}}\expandafter\ifx\expandafter\relax\detokenize{#2}\relax\else(#2)\fi$}}}
\newcommand{\dEqnSys}[2][]{\textup{\textsc{eqn${}^\ast_{\mathbf{#1}}\expandafter\ifx\expandafter\relax\detokenize{#2}\relax\else(#2)\fi$}}}

\newcommand{\prob}[1]{\textup{\textsc{#1}}\xspace}
\newcommand{\dUGAP}{\prob{ugap}}
\newcommand{\dNCL}{\prob{ncl}}
\newcommand{\dIEmpty}[1]{\prob{{#1}-int-empty}}

\newcommand{\dEMemb}[2][]{\ensuremath{E}\textnormal{-}\dMemb[#1]{#2}}
\newcommand{\dEConj}[2][]{\ensuremath{E}\textnormal{-}\dConj[#1]{#2}}
\newcommand{\dEMembS}[2][]{\ensuremath{E}\textnormal{-}\dMembS[#1]{#2}}
\newcommand{\dEConjS}[2][]{\ensuremath{E}\textnormal{-}\dConjS[#1]{#2}}

\newcommand{\mysubparagraph}[1]{\vspace*{-2mm}\subparagraph*{#1}}

\title{Membership and Conjugacy in Inverse Semigroups}
\titlerunning{Membership and Conjugacy in Inverse Semigroups}

\author{Lukas Fleischer}{FMI, University of Stuttgart \\ Universitätsstraße 38, 70569 Stuttgart, Germany}{lfleischer@lfos.de}{https://orcid.org/0000-0001-5234-4348}{}
\author{Florian Stober}{FMI, University of Stuttgart \\ Universitätsstraße 38, 70569 Stuttgart, Germany}{florian.stober@fmi.uni-stuttgart.de}{https://orcid.org/0000-0002-5516-6660}{}
\author{Alexander Thumm}{University of Siegen \\ Hölderlinstraße 3, 57076 Siegen, Germany}{alexander.thumm@uni-siegen.de}{https://orcid.org/0009-0005-4240-2045}{}
\author{Armin Weiß}{FMI, University of Stuttgart \\ Universitätsstraße 38, 70569 Stuttgart, Germany}{armin.weiss@fmi.uni-stuttgart.de}{https://orcid.org/0000-0002-7645-5867}{Supported by the German Research Foundation (DFG) grant WE 6835/1-2.}
\authorrunning{L.~Fleischer, F.~Stober, A.~Thumm, A.~Weiß}

\hideLIPIcs

\Copyright{Lukas Fleischer, Florian Stober, Alexander Thumm, Armin Weiß}

\begin{document}

\maketitle

\begin{abstract}
  The membership problem for an algebraic structure asks whether a given element is contained in some substructure, which is usually given by generators.
  In this work we study the membership problem, as well as the conjugacy problem, for finite inverse semigroups. 
The closely related membership problem for finite semigroups has been shown to be \PSPACE-complete in the transformation model by Kozen (1977) and \NL-complete in the Cayley table model by Jones, Lien, and Laaser (1976). 
In the partial bijection model, the membership and the conjugacy problem for finite inverse semigroups were shown to be \PSPACE-complete by Birget and Margolis (2008) and by Jack (2023).

Here we present a more detailed analysis of the complexity of the membership and conjugacy problems parametrized by varieties of finite inverse semigroups.
We establish dichotomy theorems for the partial bijection model and for the Cayley table model.
In the partial bijection model these problems are in \NC (resp.\ \NP for conjugacy) for strict inverse semigroups and \PSPACE-complete otherwise.
In the Cayley table model we obtain general \LOGSPACE-algorithms as well as \NPOLYLOGTIME upper bounds for Clifford semigroups and \LOGSPACE-completeness otherwise.

Furthermore, by applying our findings, we show the following:
the intersection non-emptiness problem for inverse automata is $\PSPACE$-complete even for automata with only two states;
the subpower membership problem is in \NC for every strict inverse semi\-group and \PSPACE-complete otherwise;
the minimum generating set and the equation satisfiability problems are in \NP for varieties of finite strict inverse semigroups and \PSPACE-complete otherwise.
\end{abstract}

\newpage
\thispagestyle{empty}
\tableofcontents

\newpage
\setcounter{page}{1}
\section{Introduction}

In this work, we study the membership problem  for inverse semigroups and some related problems such as the conjugacy problem.
The \emph{membership problem} for semigroups in the \emph{transformation model} has first been studied by Kozen \cite{koz77} in 1977.
 It receives as input a list $(u_1, \dots, u_k)$ of functions $u_i \colon \Omega \to \Omega$  for some finite set $\Omega$ and a target function $t \colon \Omega \to \Omega$; the question is whether $t$ can be written as composition of the $u_i$ or, with other words, whether $t $ is contained in the subsemigroup generated by $\{u_1, \dots, u_k\}$.
It is closely related to the DFA intersection non-emptiness problem, which receives as input a list of deterministic finite automata (DFAs) and asks whether there is a word accepted by \emph{all} of the automata.
Indeed, Kozen \cite{koz77} showed that both problems are \PSPACE-complete.

\medbreak

Inverse semigroups have been first studied by Wagner~\cite{Wagner52} and Preston~\cite{Preston54} to describe partial symmetries.
They constitute the arguably most natural class of algebraic structures containing groups and being contained in the semigroups.
An \emph{inverse semigroup} is a semigroup equipped with an additional unary operation $x \mapsto \ov x$ such that $x \ov x x = x$ and $ \ov x x \ov x = \ov x$ for all $x$ and $\ov x$ is unique with that property.
This clearly generalizes the inverse operation in groups.
Similar to groups being an algebraic abstraction of symmetries and semigroups being an algebraic abstraction of computation, inverse semigroups abstract \emph{symmetric computation}.
This notion of computation, where every computational step is invertible, was introduced by Lewis and Papadimitriou~\cite{LewisPapadimitriou82} in order to better describe the complexity of the accessibility problem in undirected graphs \dUGAP{}, which only much later was shown to be in \LOGSPACE{} (deterministic logspace) by Reingold~\cite{Reingold08}. 

In the setting of inverse semigroups, it is natural to consider the \emph{partial bijection model} for the membership problem, where the $u_i$ and $t$ are \emph{partial} functions which are injective on their domain.
The membership problem for inverse semigroups in the partial bijection model is also \PSPACE-complete~-- and, thus, as difficult as for arbitrary semigroups~-- as observed by Birget and Margolis \cite{BirgetM08} and independently by Jack \cite{Jack23}.
This observation is based on a result by Birget, Margolis, Meakin, and Weil \cite{BirgetMMW94} showing that the intersection non-emptiness problem for \emph{inverse automata} is \PSPACE-complete.
Roughly speaking, an inverse automaton is a DFA with a partially defined transition function where every letter induces a partial bijection on the set of states and the action of each letter can be ``inverted'' by a sequence of letters. 
Thus, inverse automata can be seen as a generalization of permutation automata, for which the intersection non-emptiness problem is \NP-complete \cite{BlondinKM12}.
Interestingly, the corresponding membership problem for permutation groups is even in \NC as shown by Babai, Luks, and Seress \cite{BabaiLS87} (for a series of preliminary results, see \cref{sec:related-work}).

\medbreak

A different variant of the membership problem has been introduced by Jones, Lien, and Laaser \cite{JonesLL76} in 1976: for the membership problem  in the \emph{Cayley table model} the ambient semigroup is given as its full multiplication table (a.k.a.\ Cayley table), \ie instead of the finite set $\Omega$ as above, the input includes a multiplication table of a semigroup and the elements are given as indices to rows/columns of that multiplication table. 
Clearly, this is a much less compressed form than the membership problem in the transformation or partial bijection model and, indeed, the membership problem  in the Cayley table model is \NL-complete \cite{JonesLL76}.

Like in the transformation model, the case of groups appears to be easier than the general case.
Indeed, in 1991, Barrington and McKenzie \cite{BarringtonM91} observed that the membership problem for groups in the Cayley table model (which they denote by ``\textsc{gen}(groups)'' and we by $\dMemb[CT]{\vG}$) can be solved in \L with an oracle for \dUGAP and speculated about it potentially being $\LOGSPACE^\dUGAP$-complete.
Indeed, they posed the following question.

\begin{quote}
	\em	Does \textup{\textsc{gen}(groups)} belong to \LOGSPACE? We doubt that this is the case: we believe rather that \textup{\textsc{gen}(groups)} is complete for the $\NC^1$-closure of \dUGAP, though we do not
  yet see how to apply the techniques in Cook and McKenzie (1987) to prove that \textup{\textsc{gen}(groups)} is even \LOGSPACE-hard.
\end{quote} 

This conjecture has been refuted by \ifAnonimous Fleischer\else the first author of the present article\fi~\cite{Fleischer19diss,Fleischer22} by showing that $\dMemb[CT]{\vG}$ can be solved in \NPOLYLOGTIME  (at least if we read completeness with respect to \ACz-reductions; when using $\NC^1$-reductions, the results in \cite{Fleischer19diss,Fleischer22} give only a strong indication that the conjecture does not hold).
Yet, in this work, we establish that the conjecture actually holds if we replace groups with inverse semigroups.

\medbreak

To find out in which cases the membership and conjugacy problem are easy and in which cases they are difficult, we study these problems restricted to certain varieties of finite inverse semigroups.
A \emph{variety of finite (inverse) semigroups}, frequently called a \emph{pseudovariety}, is a class of (inverse) semigroups that is closed under finite direct products, (inverse) subsemigroups, and quotients.
Important varieties of finite (inverse) semigroups are groups, semilattices, or aperiodic (inverse) semigroups.
Varieties of finite semigroups are closely linked to varieties of formal languages (\ie classes of languages enjoying natural closure properties) by Eilenberg's Correspondence Theorem \cite{eil76}.

Beaudry, McKenzie, and Thérien \cite{BeaudryMT92} classified the varieties of finite aperiodic monoids in terms of the complexity of their membership problem.
They found the following five classes: \ACz, \Ptime-complete, \NP-complete, \NP-hard, and \PSPACE-complete. 
Note that it might seem like a negligible difference whether semigroups or monoids are considered; however, the landscape of varieties of finite semigroup is much richer than the varieties of finite monoids. 
The aim of this work is to provide a similar classification for inverse semigroups.

\subsection{Our Results}\label{sec:results}

We consider the membership and conjugacy problems for inverse semigroups parametrized by a variety $\vV$.
For both problems we are given \emph{inverse} semigroups $U\leq S$ where $U$ is given by generators and $U \in \vV$.
Given an element $t \in S$, the membership problem asks whether $t \in U$.
Given elements $s, t \in S$, the conjugacy problem asks whether $\bar u s u = t$ and $s = u t \bar u$ for some $u\in U \cup \{1\}$.
Both problems are examined with respect to two models of input -- the Cayley table model and the partial bijection model.
We write $\dMemb[CT]{\vV}$, $\dConj[CT]{\vV}$, $\dMemb[PB]{\vV}$, and $\dConj[PB]{\vV}$, accordingly.
Regarding further details on the definition we refer to \Cref{sec:prelims-problems}.

Our main result regarding the Cayley table model is the following dichotomy.
Herein $\vCl$ denotes the variety of finite Clifford semigroups, which is the smallest variety containing all finite groups and semilattices.
The class \NPOLYLOGTIME comprises all problems solvable by non-deterministic random access Turing machines in time $\log^{\bigO(1)} n$; see \cref{sec:complexity}.

\begin{maintheorem}[Cayley Table Model]\label{thm:main-CT}
	Let $\vV$ be a variety of finite inverse semigroups. 
	\begin{itemize}
		\item If $\vV \sse \vCl$, then $\dMemb[CT]{\vV}$ and $\dConj[CT]{\vV}$ are in \NPOLYLOGTIME and in \LOGSPACE.
		\item If $\vV \not\sse \vCl$, then $\dMemb[CT]{\vV}$ and $\dConj[CT]{\vV}$ are \LOGSPACE-complete.
	\end{itemize}
\end{maintheorem}

In particular, both problems are in $\NPOLYLOGTIME \sse \qACz$ if and only if $\vV \sse \vCl$ as the class $\NPOLYLOGTIME$ contains no problem that is hard for $\LOGSPACE$ (with respect to $\ACz$ reductions).
% Here \qACz denotes the class of problems decidable by circuits of quasipolynomial size and constant depth.
Furthermore, \cref{thm:main-CT} establishes Barrington and McKenzie's conjecture \cite{BarringtonM91} on $\LOGSPACE^\dUGAP$-completeness\footnote{Recall that $\LOGSPACE^\dUGAP = \LOGSPACE$ by Reingolds seminal result \cite{Reingold08}.} of the membership problem~-- however, for the larger class of inverse semigroups or, more specifically, any variety of finite inverse semigroups not contained in $\vCl$.

The condition in \Cref{thm:main-CT} can be equivalently formulated using the following fact.
A variety of finite inverse semigroups is contained in $\vCl$ if and only if it does not contain the combinatorial Brandt semigroup $B_2$.
The latter consists of elements $\{s,\ov s, s\ov s, \ov s s, 0\}$ where $s^2 = \ov s^2 = 0$ and all of the other products are as one would expect.
Hence, by \cref{thm:main-CT}, the problems $\dMemb[CT]{\vV}$ and $\dConj[CT]{\vV}$ are \LOGSPACE-complete if and only if $B_2 \in \vV$.

\medbreak

We now turn to the partial bijection model.
This input model is similar to the transformation model for semigroups considered above~-- however, the generators and the target elements are partial maps that need to be injective on their domain. 
Our main result for the partial bijection model is the following dichotomy.

\begin{maintheorem}[Partial Bijection Model]\label{thm:main-PB}
	Let $\vV$ be a variety of finite inverse semigroups. 
	\begin{itemize}
		\item If $\vV \sse \vSI$, then $\dMemb[PB]{\vV} $ is in \NC and $\dConj[PB]{\vV}$ is in \NP.
		\item If $\vV \not\sse \vSI$, then $\dMemb[PB]{\vV}$ and $\dConj[PB]{\vV}$ are \PSPACE-complete.
	\end{itemize}
\end{maintheorem}

Herein $\vSI$ denotes the variety of \emph{strict inverse semigroups}, which is the smallest variety containing all groups and the combinatorial Brandt semigroup $B_2$; it contains $\vCl$, in particular, all semilattices (denoted as $\vSl$), and the variety generated by $B_2$ (denoted as $\vBS$).
As such, $B_2$ no longer serves as a key obstruction to an easy membership problem (as was the case in the Cayley table model).
This r\^{o}le is now played by the combinatorial Brandt monoid $B_2^1$.
Indeed, a variety of finite inverse semigroups $\vV$ is contained in $\vSI$ if and only if $B_2^1 \not\in \vV$.

The case $\vV \sse \vSI$ in \Cref{thm:main-PB} can be further refined as follows.
\begin{itemize}
	\item If $\vV \sse \vSl$, then $\dMemb[PB]{\vV} $ and $\dConj[PB]{\vV}$ are in \ACz (see \cite{BeaudryMT92}).
	\item If $\vV = \vBS$, then $\dMemb[PB]{\vV} $ and $\dConj[PB]{\vV}$ are \LOGSPACE-complete.
	\item If $\vV \not\sse \vBS$, then $\dMemb[PB]{\vV} $ is in \NC and $\dConj[PB]{\vV}$ is in \NP; both are hard for \LOGSPACE.
\end{itemize}

Note that here \ACz and $\NC$ refer to uniform circuit classes and, as such, the three levels of complexity $\ACz \subsetneq \NC \subsetneq \PSPACE$ are separated unconditionally.
Therefore, in particular, each of the problems $\dMemb[PB]{\vV}$ and $\dConj[PB]{\vV}$ is in $\ACz$ if and only if $\vV \sse \vSl$, and the problem $\dMemb[PB]{\vV}$ is in $\NC$ if and only if $\vV \sse \vSI$.

For $\vV \sse \vSI$ we reduce the problems $\dMemb[PB]{\vV}$ and $\dConj[PB]{\vV}$ to the corresponding problems for the variety of finite groups, matching their complexity.
We build on the celebrated result of Babai, Luks, and Seress~\cite{BabaiLS87} which states that the membership problem for permutation groups is in \NC, as well as the observation that, due to groups admitting polylogarithmic SLPs, the corresponding conjugacy problem is in \NP.

\medbreak

As outlined above, membership problems are deeply intertwined with intersection non-emptiness problems for the corresponding classes of automata.
In 2016, Bulatov, Kozik, Mayr, and Steindl \cite{BulatovKMS16} proved that the intersection non-emptiness problem for DFAs remains \PSPACE-complete even if the input automata are restricted to at most three states exactly one of which is accepting.
This corresponds to semigroups in the transformation model.
Here we obtain a similar result for inverse automata, which corresponds to inverse semigroups in the partial bijection model, using the same reduction as for the hardness part of \cref{thm:main-PB}.

\begin{maincorollary}\label{cor:main-intersection}
	 The intersection non-emptiness problem for inverse automata is \PSPACE-complete.
   This holds even if the automata have only two states, one of which is accepting.
\end{maincorollary}

The reason that two states suffice to show \PSPACE-hardness is grounded in the fact that inverse automata have partially defined transition functions, whereas the above-mentioned result concerns automata with total transition functions.

In the same work Bulatov, Kozik, Mayr, and Steindl also considered the \emph{subpower membership problem} and showed that it is \PSPACE-complete for arbitrary semi\-groups.
Let $S$ be a fixed (inverse) semigroup.
The input for the subpower membership problem for $S$ consists of a number $m$, elements $u_1, \dots, u_k \in S^m$, and $t\in S^m$.
The question is whether $t $ is contained in the (inverse) subsemigroup generated by $\{u_1, \dots, u_k\}$.
As a consequence of \cref{thm:main-PB} and \cref{cor:main-intersection}, we obtain the following dichotomy for the complexity of the subpower membership problem for inverse semigroups.

\begin{maincorollary}\label{cor:main-subpower}
The subpower membership problem for an inverse semigroup $S$ is in~\NC if and only if $S \in \vSI$.
		Otherwise, it~is~\PSPACE-complete.	
\end{maincorollary}

Finally, we apply our results to the problems of determining the minimal size of a generating set (\dMGS{}) and deciding satisfiability of an equation (\dEqn{}) and obtain a similar dichotomy as above.
The minimum generating set problems receives as input an inverse semigroup $S$ and an integer $k$ and asks whether $S$ can be generated by at most $k$ elements.
For \dEqn{} the input is an inverse semigroup $S$ and a single equation, and the question is whether there is a satisfying assignment of the variables to elements of $S$.

\begin{maincorollary}\label{cor:main-mgs-equations}
	Let $\vV$ be a variety of finite inverse semigroups.
  \begin{itemize}
    \item If $\vV \sse \vSI$, then $\dMGS[PB]{\vV}$ and $\dEqn[PB]{\vV}$ are in \NP.
    \item If $\vV \not\sse \vSI$, then $\dMGS[PB]{\vV}$ and $\dEqn[PB]{\vV}$ are \PSPACE-complete.
  \end{itemize}
\end{maincorollary}

\subsection{Technical Overview}\label{sec:technical-overview}

A central role in our results is played by the combinatorial Brandt semigroup $B_2$ (defined above) and the Brandt monoid $B_2^1$, which is the Brandt semigroup with an adjoined identity.
The former and the latter are the sole obstruction to inclusion in the variety $\vCl$ and $\vSI$, respectively.
As such, both inverse semigroups are crucial obstructions preventing the membership and conjugacy problems from being ``easy'':
for example in the Cayley table model, the Brandt semigroup $B_2$ is the obstruction from $\dMemb[CT]{\vV}$ being in \NPOLYLOGTIME; in the partial bijection model the Brandt monoid $B_2^1$ makes the problem \PSPACE-hard.

\mysubparagraph{Outline of the Proof of \cref{thm:main-CT}.}
The proof of \cref{thm:main-CT} consists of three main steps: first, show that for Clifford semigroups both problems can be solved in \NPOLYLOGTIME, second, show that in any case they can be solved in \LOGSPACE, and finally, show that, if the variety under consideration contains the Brandt semigroup $B_2$, then the problems are hard for \LOGSPACE.

For the first point there is not much left to do.
Indeed, \ifAnonimous Fleischer\else the first author\fi~\cite{Fleischer19diss,Fleischer22} showed that $\dMemb[CT]{\vCl}$ is in \NPOLYLOGTIME.
Thus, it merely remains to apply this result also to the conjugacy problem.
Yet, for the sake of completeness and because the proof allows us to describe some interesting consequences, we give a proof in \Cref{sec:Clifford-CT}.
The crucial idea is to use compression via straight-line programs (SLPs) and the Reachability Lemma due to Babai and Szemer\'edi \cite{BabaiS84}:
if $G$ is a finite group and $g \in G$, then there is an SLP of length $\bigO(\log^2 \abs{G})$ that computes $g$.
We also say that groups \emph{admit polylogarithmic SLPs}.
One can guess such an SLP and verify whether it actually computes $g$ in \NPOLYLOGTIME.
In \cite{Fleischer19diss,Fleischer22} \ifAnonimous Fleischer\else the first author\fi\ extended the Reachability Lemma to Clifford semigroups.

Conversely and as a consequence of both parts of \Cref{thm:main-CT} we also obtain a complete characterization of when a variety of finite inverse semigroups  $\vV$ admits polylogarithmic SLPs, namely this is the case if and only if $\vV \sse \vCl$ (see \cref{cor:short-SLP-char}).

The proof that $\dMemb[CT]{\vV}$ and $\dConj[CT]{\vV}$ are hard for \LOGSPACE if $\vV \not\sse \vCl$ follows via a reduction from undirected graph accessibility (\dUGAP), which is intimately related to the membership problem for the combinatorial Brandt semigroups $B_n$; for details, see \cref{sec:SL-hard}.
Finally, to prove that $\dMemb[CT]{}$ and $\dConj[CT]{}$ are in \LOGSPACE, the crucial observation is that the strongly connected components of the Cayley graph of an inverse semigroup are actually undirected graphs. 
This allows to reduce the problem to $\dUGAP$, which by Reingold's result \cite{Reingold08} is in \LOGSPACE.
While for the conjugacy problem this yields a direct many-one reduction to $\dUGAP$, for the membership problem we apply a greedy algorithm using oracle calls for $\dUGAP$ in the associated Cayley graph.
To be more precise, to decide whether $t \in U$, we start initializing a variable $x$ to the neutral element, which we assume to exist and to be contained in $U$.
We update $x$ iteratively while maintaining the invariants that $x \in U$ and $x \ov x t = t$ (\ie $x$ is greater than or equal to $t$ with respect to Green's relation $\gR$). 
Using the \dUGAP oracle, in each iteration of the algorithm we greedily pick an element $x_{\mathrm{new}}$ to replace $x$ that satisfies the invariants and is adjacent to the strongly connected component of $x$ in the associated Cayley graph (so that $x$ is strictly greater than $x_{\mathrm{new}}$ with respect to $\gR$).
While the invariants guarantee that, at any point, we still might multiply $x$ on the right by another element of $U$ to reach $t$, the way we choose  $x_{\mathrm{new}}$ ensures that we actually make progress towards $t$.

\mysubparagraph{Outline of the Proof of \cref{thm:main-PB}.}

Our approach to the dichotomy result for the partial bijection model, \ie to proving \cref{thm:main-PB}, is similar to the above.
In the group case, we build on the celebrated result of Babai, Luks, and Seress \cite{BabaiLS87} which states that the membership problem for permutation groups is in \NC, as well as the observation that, due to groups admitting polylogarithmic SLPs, the corresponding conjugacy problem is in \NP.

Extending these bounds to Clifford semigroups is rather straight-forward: 
one can identify an appropriate subgroup (in fact, an $\gH$-class\footnote{Here $\gH$ and $\gD$ refer to Green's relations; for a definition, see \cref{sec:Green-conjugacy}.}) to which the problem can be reduced to; for details, see \cref{sub:clifford-pb}.
Interestingly, a similar reduction also works for strict inverse semigroups.
However, the proof is much more involved as, in that case, identifying an appropriate subgroup is no longer possible in \ACz but hard for \LOGSPACE due to the presence of Brandt semigroups.
We show that a \LOGSPACE-reduction is nonetheless possible building on some special properties of representations of strict inverse semigroups which we now briefly describe.

Suppose that $U$ is a strict inverse semigroup generated by a set $\Sigma$ of partial bijections on some set $\Omega$.
We say that an element $u \in U$ is $\Delta$-large for some $U$-invariant subset $\Delta \sse \Omega$ if its domain $\dom(u)$ or, equivalently, its range $\ran(u)$ intersects every $U$-orbit $x^U \sse \Delta$.
Consider the graph $\mathrm{M}(\Delta; \Sigma)$ which, for each $\Delta$-large $u \in \Sigma$, possesses an edge labeled $u$ from a vertex associated with the set $\Delta \cap \dom(u) \sse \Omega$ to a vertex associated with the set $\Delta \cap \ran(u) \sse \Omega$.
This graph, which we call the \emph{Munn graph} $\mathrm{M}(\Delta; \Sigma)$ at $\Delta$ and with respect to $\Sigma$, is the basis of our reduction.
As it turns out, every $\gD$-class of the strict inverse semigroup $U$ is generated (as a groupoid) by a connected component of the Munn graph $\mathrm{M}(\Delta; \Sigma)$ at some suitably chosen $U$-invariant set $\Delta \sse \Sigma$.
Moreover, crucially, we can identify the set $\Delta$ suitable for the $\gD$-class of $U$ that contains a given element $t \in U$.
This allows us to first reduce the membership problem for $U$ to some $\gD$-class of $U$ and, ultimately, to some $\gH$-class of $U$.

We refer the reader to \cref{sub:munn} for details on the properties of Munn graphs and to \cref{sub:sis} for details of the reduction outlined above as well as a corresponding reduction for the conjugacy problem based on the same ideas.
We crucially rely on the graph accessibility problem (\dUGAP) for the Munn graph.
Even more, all relevant computations can be preformed in $\LOGSPACE^{\dUGAP}$ and are, in fact, even easy to implement if one replaces oracle calls to \dUGAP with standard algorithms for this problem.
On the other hand, using $\dUGAP \in \LOGSPACE$ \cite{Reingold08} yields $\dMemb[PB]{\vSI} \leq_m^{\LOGSPACE} \dMemb[PB]{\vG}$ and $\dConj[PB]{\vSI} \leq_m^{\LOGSPACE} \dConj[PB]{\vG}$.

Finally, let us attend to the general case of the membership and conjugacy problems for inverse semigroups in the partial bijection model.
On the one hand, it is well-known that both problems can be solved in \PSPACE. 
On the other hand, both problems are \PSPACE-hard in general \cite{BirgetM08,Jack23}. %, \ie without any restrictions imposed on the inverse semigroups in question.
%Their proofs are based on the \PSPACE-hardness result for the intersection non-emptiness problem for inverse automata due to Birget, Margolis, Meakin, and Weil \cite{BirgetMMW94}.
Here we show that the (idempotent) membership and conjugacy problems are \PSPACE-hard for any variety of finite inverse semigroups $\vV$ containing the combinatorial Brandt monoid $B^1_2$ or, equivalently, $\vV \not\sse \vSI$.
We do so via reduction from \dNCL, the configuration-to-configuration problem variant of non-deterministic constraint logic (NCL). 
This problem, introduced and shown to be \PSPACE-complete by Hearn and Demaine~\cite{HearnD05}, asks whether two given configurations of an NCL machine can be transformed into one another.
Crucially, in this problem, configurations and transitions between these can be specified locally.
We encode each local aspect of a problem instance into a (small) combinatorial Brandt monoid $B^1_n$ and the entire instance into the Cartesian product of all  these $B^1_n$; for details, we refer the reader to \cref{sub:pbm-hardness}.
Here we use of the fact that $B^1_n$ divides the $n$-fold Cartesian power $(B^1_2)^n$ and, thus, $B^1_n \in \vV$ whenever $B^1_2 \in \vV$.

The hardness proof for \cref{cor:main-intersection} closely follows the hardness proof of \cref{thm:main-PB}, but using Cartesian powers of $B^1_2$ directly. 
Then we obtain \cref{cor:main-subpower}, the \PSPACE-hardness of the subpower membership problem, as a corollary of \cref{cor:main-intersection} by simply replacing the combinatorial Brandt monoid $B^1_2$ with any inverse semigroup $S$ divided by it.

\mysubparagraph{Details on Further Results.}

Our results on the minimum generating set problem and on deciding satisfiability of equations in \cref{cor:main-mgs-equations} are rather direct applications of our main results. 
In both cases, the upper bounds are established via an algorithm that guesses a witness or solution, which is then verified in polynomial time using access to an oracle for the membership problem.

The lower bounds, \ie \PSPACE-hardness, are obtained via reductions from (suitably restricted variants) of the membership and conjugacy problems.
In the case of the minimum generating set problem, we reduce from (such a variant) $\dMembS[PB]{\vV}$ to $\dMGS[PB]{\vV \vee \vSl}$ using the following idea.
Given an instance $\Sigma \sse \ISym(\Omega)$ and $t\in \ISym(\Omega)$ of the former problem, we would like that the inverse subsemigroup $\gen{\Sigma \cup \{t\}}$ is generated by $\abs{\Sigma}$ elements if and only if $t \in \gen{\Sigma}$.
However, this is clearly not the case as $\Sigma$ might contain redundant generators.
Nevertheless, by adding extra elements to $\Omega$ for each element of $\Sigma$ on which the generators behave as a semilattice, we can ensure that every single generator of $\Sigma$ is needed; 
thus, in this modified instance the above wishful thinking actually applies.

To see \PSPACE-hardness of deciding satisfiability of equations, observe first that conjugacy is represented by a system of equations.
As such, the problem $\dEqnSys[PB]{\vV}$ of deciding whether a system of equations has a solution is \PSPACE-hard whenever $\vV \not\sse \vSI$ by \cref{thm:main-PB}.
Using a suitably restricted variant $\dEConjS[PB]{\vV}$ of the conjugacy problem, were a single equation suffices to represent conjugacy, we obtain that the corresponding problem \dEqn[PB]{\vV} of deciding whether a single equation has a solution is also \PSPACE-hard whenever $\vV \not\sse \vSI$.

\subsection{Related Work}\label{sec:related-work}

Inverse semigroups have been introduced by Wagner~\cite{Wagner52} and Preston~\cite{Preston54} to formalize partial symmetries. 
Implicitly, they had been studied even before for example in the context of so-called \emph{pseudogroups} \cite{Golab39}.
Inverse semigroup have been investigated extensively from geometric, combinatorial, and algorithmic viewpoints; see \eg \cite{ElliottLM24,Gray2020,Margolis-Meakin:1989,ms96,mun74,OlijnykSS10,Djadchenko77,Kleiman76,Kleiman79} for a rather random selection.
For additional background on inverse semigroups, we refer to the standard books \cite{Law99,Petrich84} and the many references therein.

While membership problems in infinite semigroups recently have also gained a lot of attention (see \cite{ChoffrutK05,BellHP17,DiekertPS2024,Dong24} for a few examples), let us in the following give an overview on related work on the membership problem with the input models we use in the present work.

\mysubparagraph{Membership Problem (Cayley table model).}

The membership problem has been studied for many algebraic structures.
Indeed \cite{JonesL76} shows that the membership problem for magmas (\ie\ having a binary operation with no additional axioms) is \Ptime-complete in the Cayley table model.
In contrast, by \cite{CollinsGLW24} for quasigroups (magma with ``inverses'', a.k.a.\ latin squares) it is in \NPOLYLOGTIME using similar techniques as  we apply for the first part of \cref{thm:main-CT}. 

The membership problems for semigroups in the Cayley table model has been introduced by Jones, Lien, and Laaser \cite{JonesLL76}.
Further studies by Barrington, Kadau, and Lange \cite{BarringtonKLM01} showed that it can be solved in $\mathsf{FOLL}$ for nilpotent groups of constant class.
This result has been further improved by Collins, Grochow, Levet, and \ifAnonimous Weiß\else the last author\fi~\cite{CollinsGLW24} showing that the problem can be solved in $\mathsf{FOLL}$ for all nilpotent groups (\ie of arbitrary class) and served as a catalyst for \ifAnonimous Fleischer\else the first author\fi's work \cite{Fleischer19diss,Fleischer22} giving \NPOLYLOGTIME algorithms for Clifford semigroups, on which we build in the present work.

\mysubparagraph{Membership Problem (Transformation/Partial Bijection Model).}

The membership problem for semigroups in the transformation model has been shown to be $\PSPACE$-complete by Kozen~\cite{koz77}.
Beaudry \cite{Beaudry88} showed that the membership problem in commutative semigroups is \NP-complete in the transformation model.
This was later extended in \cite{Beaudry88thesis,BeaudryMT92} to classify the complexity of the membership problem in aperiodic monoids as outlined above.

Based on Sims' work \cite{Sims67}, Furst, Hopcroft, and Luks \cite{FurstHopcroftLuks80} showed that the membership problem for permutation groups is solvable in polynomial time, which after several partial results \cite{LuksM88,Luks86,McKenzieC87} was improved to \NC by Babai, Luks, and Seress \cite{BabaiLS87}.
Interestingly, the problem of rational subset membership is \NP-complete due to Luks \cite{Luks93} (see also \cite{LohreyRZ22}).

Turning our attention to the partial bijection model, it was observed by Birget and Margolis \cite{BirgetM08} and, recently, by Jack \cite{Jack23} that the membership problem for \emph{inverse semigroups} given by partial bijections is \PSPACE-complete.
This follows from an earlier result by by Birget, Margolis, Meakin, and Weil~\cite{BirgetMMW94} showing that the intersection non-emptiness problem for \emph{inverse automata} is \PSPACE-complete.
Indeed, this problem remains \PSPACE-complete over a two-letter alphabet \cite{BirgetM08}.

\mysubparagraph{Subpower Membership Problem.}
While there is no obvious generalization of the partial bijection or transformation semigroup model to non-associative structure such as magmas, the subpower membership problem still can be posed in this case.
Indeed, the subpower membership problem initially has been studied within the context of universal algebra, see \eg \cite{Mayr12,BulatovMS19,Kompatscher24}, and has turned out to be $\mathsf{EXPTIME}$-complete \cite{Kozik08} in general.
For arbitrary semi\-groups the subpower membership problem has been shown to be $\PSPACE$-complete by Bulatov, Kozik, Mayr, and Steindl \cite{BulatovKMS16}.

Further results on the subpower membership problem in semigroups are due to Steindl giving a \Ptime vs.\ \NP-completeness dichotomy for the special case of bands \cite{Steindl17} and a \Ptime vs.\ \NP-complete vs.\ \PSPACE-complete trichotomy for combinatorial Rees matrix semigroups with adjoined identity
\cite{Steindl19}. 
Here it is interesting to note that by our results the \NP-completeness case does not exists for inverse semigroups.

\medbreak

We now turn our attention to the intimately related intersection non-emptiness problem.

\mysubparagraph{Intersection Non-Emptiness Problem.}
The DFA intersection non-emptiness problem has been introduced and shown to be $\PSPACE$-complete by Kozen~\cite{koz77}.
Further work studying the complexity (including parametrized and fine-grained complexity) of the DFA intersection non-emptiness problem can be found in \cite{FernauHW21,KarakostasLV03,LangeR92,HolzerK11,SwernofskyW15,Wehar14,OliveiraW20,ArrighiFHHJOW21}.
Two special cases are that the DFA intersection non-emptiness problem is \NP-complete for DFA accepting finite languages \cite{RampersadS10} and for DFA over a unary alphabet \cite{StockmeyerM73} (see also \cite{FernauK17}).

Another important special case are permutation automata~\cite{Thierrin68} (a.k.a.\ group DFAs). 
This is closely linked to the membership problem in groups, which is in \NC \cite{BabaiLS87}. 
Thus, it comes rather as a surprise that the intersection non-emptiness problem is \NP-complete as Brondin, Krebs, and McKenzie \cite{BlondinKM12} showed; however, when restricting to permutation automata with a single accepting state it, indeed, is in \NC \cite{BlondinKM12}.
Even more, intersection non-emptiness for permutation automata plus one context-free language is \PSPACE-complete \cite{LohreyRZ22}.

Note that every permutation automaton is an inverse automaton as studied in the present work and \eg by Birget, Margolis, Meakin, and Weil \cite{BirgetMMW94}.
Furthermore, inverse automata are a special case of \emph{reversible} automata (or injective automata as they are called in \cite{BirgetMMW94}), which were studied \eg by Pin \cite{Pin87} and Radionova and Okhotin \cite{RadionovaO24}.

\medbreak

Another problem related to membership is the conjugacy problem, which for infinite groups was introduced by Dehn in 1911 \cite{Dehn11}.
For generalizations to semigroups see \cite{Otto84}.

\mysubparagraph{Conjugacy Problem.}
The conjugacy problem for permutation groups is in \NP and hard for graph isomorphism as shown by Luks \cite{Luks93}.
Jack \cite{Jack23} showed that the conjugacy problem for inverse semigroups in the partial bijection model is \PSPACE-complete. 
For a overview on different variants of conjugacy in (inverse) semigroups, we refer to \cite{AraujoKinyonKnieczny19}.

\medbreak

The following problems, which are closely tied to the membership and conjugacy problems (see, e.g., \cref{lem:mgs-reduction} and \cref{prop:eqn-alg}), have also attracted independent interest.

\mysubparagraph{Minimum Generating Set Problem.}

The minimum generating set problem has first been considered by Papadimitriou and Yannahakis \cite{PapadimitriouY96} and further studied in \cite{ArvindT06,Tang13Thesis} showing polylogarithmically space-bounded algorithms.
For groups, it has been shown recently to be solvable in polynomial time by Lucchini and Thakkar \cite{LucchiniT24}.
This bound was further improved to \NC by Collins, Grochow, Levet, and \ifAnonimous Weiß\else the last author\fi~\cite{CollinsGLW24}.
Moreover, they also showed that the minimum generating set problem for magmas is \NP-complete.

\mysubparagraph{Equations.}

There is an extensive work on equations in algebraic structures. 
In particular, the case of groups has attracted a lot of attention after Goldmann and Russell \cite{GoldmannR02} showed that deciding satisfiability of a system of equations is \NP-complete for every fixed non-abelian group and in \Ptime for abelian groups. 
For more recent conditional lower bounds and algorithms for deciding satisfiability of (single) equations, see \eg \cite{IdziakKKW22,IdziakKKW24,FoldvariH19}.

The case of semigroups has attracted much less attention.
While here the closely related problem of checking identities has been investigated thoroughly, \cite{Klima09,SeifS06,Kisielewicz04,AlmeidaVG09,Seif05}, there is relatively little work on deciding whether a (system of) equation(s) has a solution.

In \cite{BarringtonMMTT00} the problem of deciding whether a (single) equation in finite monoids is satisfiable has been investigated.
Among other results it has been shown that in the Brandt monoid $B_2^1$, which also plays an important role in our work, this problem is \NP-complete.
Furthermore, in \cite{KlimaTT07} systems of equations in semigroups were studied.
They presented dichotomy results for the class of finite monoids and the class of finite regular semigroups.
The result for finite regular semigroups is for a restricted variant of the problem, where one side of each equation contains no variable.

\subsection{Outline}

After fixing our notation in \cref{sec:prelims}, we first consider the special cases of groups and Clifford semigroups in \cref{sec:groups} and \cref{sec:clifford}, respectively.
After that, we turn our attention to strict inverse semigroups and prove that, like for Clifford semigroups, their membership and conjugacy problems can be reduced to the respective problems for groups.

In \cref{sec:SL-Cayley}, we present and prove our results on the Cayley table model, \ie the \LOGSPACE-completeness statement in \cref{thm:main-CT}.
In \cref{sec:PSPACE-hard} we give our \PSPACE-hardness proof (part of \cref{thm:main-PB}) and discuss the consequences for the intersection non-emptiness problem for inverse automata and the subpower membership problem.
In \cref{sec:applications} we apply our results to the minimum generating set problem and the problem of solving equations.
Finally, in \cref{sec:discussion} we provide a short summary of our results and discuss interesting open questions.

\section{Preliminaries and Notation}\label{sec:prelims}

\subsection{Inverse Semigroups}\label{sec:prelims-inverse}

A \emph{semigroup} is a non-empty set equipped with an associative binary
operation denoted by $xy$.
For a semigroup $S$ we write $E(S)$ to denote its set of \emph{idempotents}, \ie elements $e \in S$ that satisfy $e e = e$.
A  monoid is a semigroup $M$ with a \emph{neutral element}, \ie{}an element $1 \in M$ such that $1x = x = x1$ for all $x \in M$, which throughout is denoted as $1$.
A \emph{zero element} $z$ of $S$ satisfies $zx = z = xz$ for all $x \in S$; if it exists, we denote it by $0$.
For further background on general semigroups we refer to  \cite{cp67,alm94,rs09qtheory} and for inverse semigroups to \cite{Petrich84,Law99}.

An \emph{inverse semigroup} is a semigroup $S$ where every element $x \in S$  possesses a \emph{unique} inverse $\ov x \in S$, i.e., $x \ov x x = x$ and $\ov x x \ov x = \ov x$ hold and $\ov x$ is the unique element with this property.
In an inverse semigroup $S$ all idempotents commute; in particular, $E(S)$ is always a subsemigroup of $S$ and a semilattice.
We denote the \emph{natural order} on the elements of an inverse semigroup by $\leq$, \ie $x \leq y$ if and only if $x =  x \ov x y$ or, equivalently, $x = y \ov x x$.

For an inverse semigroup $S$ and a subset $\Sigma$ of $S$, we denote by $\gen{\Sigma}$ the inverse subsemigroup of $S$ \emph{generated by $\Sigma$}, \ie the smallest inverse subsemigroup of $S$ containing $\Sigma$. 
The elements of the set $\Sigma$ are called \emph{generators} for $\gen{\Sigma}$.
Note that all elements of $\gen{\Sigma}$ can be written as words over $\Sigma \cup \ov{\Sigma}$.
Therefore, we assume from now on that generating sets are closed under formation of inverses, i.e., $\ov\Sigma = \Sigma$.
However, be aware that unlike in finite groups, arbitrary subsemigroups of an inverse semigroup need not be inverse semigroups again. 

An inverse semigroup $T$ is a \emph{divisor} of $S$, written $T \preccurlyeq S$, if there exists a surjective homomorphism from an inverse subsemigroup of $S$ onto $T$.

\mysubparagraph{Symmetric Inverse Semigroups.}

The \emph{symmetric inverse semigroup $\ISym(\Omega)$} is the inverse semigroup of all partial bijections on a set $\Omega$, \ie partial maps $s\colon\Omega \to \Omega$ that induce a bijection from their \emph{domain} $\dom(s)$ to their \emph{range} $\ran(s)$.
We write $\ISym_n = \ISym(\{1, \dotsc, n\})$.

For $s \in \ISym(\Omega)$ and $x \in \Omega$, we write $x^s$ for the image of $x$ under $s$, where $x^s = \bot$ means that $x^s$ is undefined.
We extend this notation to sets $\Delta \sse \Omega$, i.e., $\Delta^s = \{ x^s \mid x \in \Delta\}$.
Be aware that $\Delta^s$ might be empty even if $\Delta$ is not (this happens if $\Delta \cap \dom(s) = \emptyset$). 
For~$\Delta \sse \Omega$ we write $e_\Delta$ for the idempotent associated to $\Delta$, which is the partial identity on $\Delta$.
Every idempotent of $\ISym(\Omega)$ is of this form, and $e_{\Delta'} e_{\Delta''} = e_{\Delta' \cap \Delta''}$.
We also write $e_{\Delta'} \vee e_{\Delta''} = e_{\Delta' \cup \Delta''}$.

The inverse subsemigroups of $\ISym(\Omega)$ are sometimes called partial bijection semigroups.
An important result by Preston~\cite{Preston54} and Wagner~\cite{Wagner52} states that every inverse semigroup $S$ can be embedded into the symmetric inverse semigroup $\ISym(S)$ in a natural way.

\mysubparagraph{Important Combinatorial Inverse Semigroups.}

We denote the two-element semilattice~by~$Y_2$, which consists of a zero $0$ and a neutral element $1$.
The Cartesian power $(Y_2)^\Omega$ is naturally isomorphic to $E(\ISym(\Omega))$; in particular, every semilattice embeds into some power of $Y_2$.

The (combinatorial) \emph{Brandt semigroup} $B(\Omega)$ on some set $\Omega$ is the inverse subsemigroup $\{ s \in \ISym(\Omega) \mid \lvert \dom(s) \rvert \leq 1 \}$ of $\ISym(\Omega)$.
The (combinatorial) \emph{Brandt monoid} on $\Omega$ is the inverse submonoid $B^1(\Omega) = B(\Omega) \cup \{1\}$ of $\ISym(\Omega)$.
We write $B_n$ and $B^1_n$ in case $\Omega = \{1,\dotsc,n\}$.

The Brandt semigroup (or monoid) can be thought of as the complete directed graph on vertices $\Omega$ together with an additional zero element (and an identity) where the multiplication of edges $(u,v)$ and $(w,z)$ is $(u,z)$ if $v = w$ and $0$ otherwise. 
For example, \[
B_2 = \big\{ \smash{\mbox{
		\begin{tikzpicture}[baseline=-3]
			\draw[thick] (0:1.4cm) circle(3pt);
			\draw[thick] (180:1.4cm) circle(3pt);  				
			\foreach \a in {0,180} {
				\draw[thick, shorten <=5pt, shorten >=5pt, ->] (\a:1.4cm) to[in=\a+30, out=\a-30, loop, min distance=.9cm] (\a:1.4cm);
				\draw[thick, shorten <=8pt, shorten >=5pt, ->] (\a:1.4cm) to[bend right=15] (\a + 180:1.4cm);
			};
		\end{tikzpicture}
}}  \big\} \cup \{0\}.
\]

\subsection{Green's Relations and Conjugacy}\label{sec:Green-conjugacy}

The following relative variants of \emph{Green's relations} will be very useful.
As usual, given an inverse semigroup $S$, we denote by $S^1$ the smallest inverse monoid containing $S$.
However, in a relative context, i.e., for an inverse subsemigroup $U \leq S$, we denote by $U^1$ the inverse submonoid $U \cup \{1\} \leq S^1$.
This abuse of notation is employed in the following definition.

\begin{definition}
	Let $U \leq S$.
	Given elements $s,t \in S$, we write 
	\begin{gather*}
		s \gLle[U] t \iff U^1 \, s \subseteq U^1 \, t, \\
		s \gRle[U] t \iff s \, U^1 \subseteq t \, U^1, \\
		s \gJle[U] t \iff U^1 \, s \, U^1 \subseteq U^1 \, t \, U^1.
	\end{gather*}
	Furthermore, we write $s \gL[U] t$ provided that $s \gLle[U] t$ and $s \gLge[U] t$; the Green's relations $\gR[U]$ and $\gJ[U]$ are defined similarly.
  Finally, we let ${\gH[U]} = {\gL[U] \cap \gR[U]}$ and ${\gD[U]} = {\gL[U] \vee \gR[U]}$.\footnote{The definition of ${\gD[U]}$ is provided here only for completeness, as ${\gD_U} = {\gJ[U]}$ whenever $S$ is finite.}
\end{definition}

We recover the usual definition of Green's relations~\cite{gre51:short} as ${\gXle} = {\gXle[S]}$ and ${\gX} = {\gX[S]}$, where ${\gX} \in \{ {\gL}, {\gR}, {\gJ} \}$ or ${\gX} \in \{ {\gH}, {\gL}, {\gR}, {\gD}, {\gJ} \}$, respectively.

\begin{lemma}\label{lem:green-semi_rel}
  Let $S$ be a finite inverse semigroup. 
  If $u s v \gJge s$ for some $s \in S$ and $u,v \in S^1$, then $\bar u u s v \bar v = s$.
  In particular, if $s \gX t$ for some $s,t \in S$ with ${\gX} \in \{ {\gL}, {\gR}, {\gJ} \}$ and $U \leq S$, then $s \gXle[U] t$ if and only if $s \gXge[U] t$.
\end{lemma}
\begin{proof}
  If $u s v \gJge s$, then $s = u' u s v v'$ for some $u', v' \in S^1$.
	As such, we have
	\begin{equation*}
    s = u' u s v v' = (u' u)^\omega s (v v')^\omega = (u' u)^\omega \bar u u s v \bar v (v v')^\omega
    = \bar u u (u' u)^\omega s (v v')^\omega v \bar v = \bar u u s v \bar v.
	\end{equation*}
  Taking $u,v \in U^1$ in the above, shows that $s \gJge[U] t \gJge s$ implies $t \gJge[U] s$.
	The remaining cases ${\gX} = {\gL}$ and ${\gX} = {\gR}$ follow by setting $v = v' = 1$ and $u = u' = 1$, respectively.
\end{proof}

Similarly, we will use a relative variant of \emph{conjugacy} in inverse semigroups defined as follows.
Beware that there are other notions of conjugacy frequently considered in the context of semigroup theory; see also \cite{AraujoKinyonKnieczny19,Jack23}.

\begin{definition}\label{def:conjugacy}
	Let $U \leq S$ be inverse semigroups.
  We call $s,t \in S$ \emph{conjugate} relative to $U$, written $s \sim_U t$, if there exists some $u \in U^1$ such that $\bar u s u = t$ and $s = u t \bar u$.
\end{definition}

Conjugacy and $\gJ$-equivalence are closely related for idempotents of inverse semigroups.

\begin{lemma}\label{lem:idem_conj_is_green_j}
	Let $U \leq S$ be inverse semigroups, and let $e, e' \in E(S)$.
	Then $e \gJge[U] e'$ holds if and only if $e' = \bar u e u$ for some $u \in U^1$.
	If $S$ is finite, then $e \gJ[U] e'$ if and only if $e \sim_U e'$.
\end{lemma}
\begin{proof}
	Suppose that $e \gJge[U] e'$, i.e., there exist $v_1, v_2 \in U^1$ with $e' = v_1 e v_2$.
	Then 
	\[
	e' = e'e' = v_1 e v_2 \bar v_2 e \bar v_1 = v_1 e v_2 \bar v_2 v_2 \bar v_2 e \bar v_1 = v_1 v_2 \bar v_2 e e v_2 \bar v_2 \bar v_1 = v_1 v_2 \bar v_2 e v_2 \bar v_2 \bar v_1 = \bar u e u
	\]
	where $u = v_2\bar v_2 \bar v_1 \in U^1$.
	The converse is trivial.
	If $S$ is finite, then $\bar u e u = e' \gJge e$ implies that $e = u \bar u e u \bar u = u e' \bar u$ by \cref{lem:green-semi_rel} and, hence, that $e \sim_U e'$.
\end{proof}

\subsection{(Pseudo-)Varieties}

A class $\cC$ of finite inverse semigroups is called a \emph{variety of finite inverse semigroups} (a.k.a.\ \emph{pseudovariety}) if it is closed under formation of finite direct products and of divisors.
By a theorem of Reiterman~\cite{rei82:short}, such a class $\cC$ consists of all finite inverse semigroups that satisfy some set of (pseudo-)identities.
We also note that these classes are intimately related to certain classes of formal languages through Eilenberg's Correspondence Theorem \cite{eil76}.

We will use boldface roman letters to denote varieties of finite inverse semigroups and sometimes also give a set of defining inverse semigroup identities for them.
For example, $\vG = \vId{x\ov x = y \ov y} =  \vId{x \bar x = 1}$ denotes the variety of finite groups, $\vCl = \vId{x \bar x = \bar x x}$ denotes the variety of finite \emph{Clifford semigroups}, and $\vSI = \vId{\bar y x y \bar y \bar x y = \bar y \bar x y \bar y x y}$ denotes the variety of finite \emph{strict inverse semigroups}.
An overview of the varieties of finite inverse semigroups most relevant to our work is given in \cref{fig:varieties}.

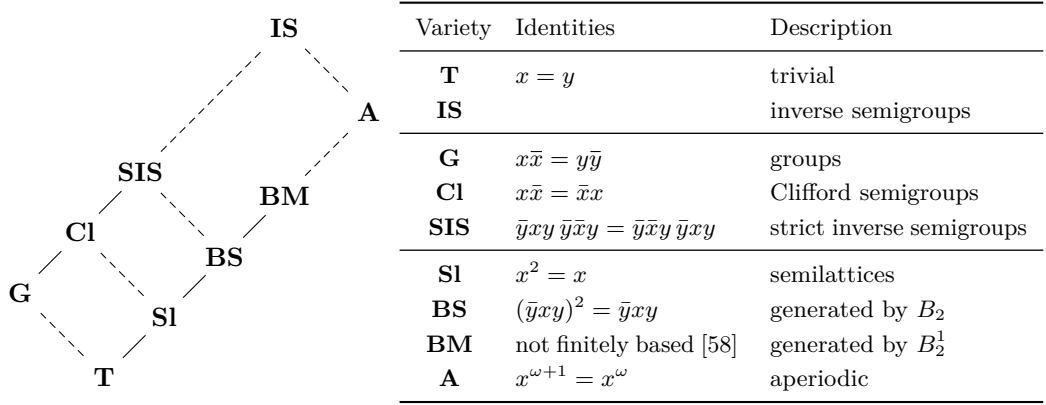
\begin{figure}[ht]
\begin{minipage}{.37\textwidth}
  \begin{tikzpicture}[scale=.795]
    % aperiodic varieties
    \foreach \v [count=\n] in {\vT, \vSl, \vBS, \vBM} {
      \coordinate (a\n) at (\n, \n);
      \node at (a\n) {\(\v\)};
    }
    \coordinate (aa) at (5.4, 5.4);
    \node at (aa) {\(\mathbf{A}\)};

    % general varieties
    \foreach \v [count=\n] in {\vG, \vCl, \vSI} {
      \coordinate (g\n) at (\n - 1.4, \n + 1.4);
      \node at (g\n) {\(\v\)};
    }
    \coordinate (ga) at (4, 6.8);
    \node at (ga) {\(\mathbf{IS}\)};

    % coverings
    \begin{scope}[shorten <= 10pt, shorten >= 10pt]
      \draw (a1) -- (a2);
      \draw (a2) -- (a3);
      \draw (a3) -- (a4);

      \draw (g1) -- (g2);
      \draw (g2) -- (g3);
    \end{scope}

    % inclusions
    \begin{scope}[densely dashed, shorten <= 11pt, shorten >= 11pt]
      \draw (a1) -- (g1);
      \draw (a2) -- (g2);
      \draw (a3) -- (g3);

      \draw (aa) -- (ga);

      \draw (a4) -- (aa);
      \draw (g3) -- (ga);
    \end{scope}
  \end{tikzpicture}
\end{minipage}
\begin{minipage}{.6\textwidth}
  \centering
  \newcolumntype{D}[1]{>{\centering\arraybackslash}m{#1}}
  \begin{tabular}[t]{p{.9cm}p{2.95cm}p{3.4cm}}
    \toprule
    Variety & Identities & Description \\
    \midrule
    \hfil$\vT$& $x = y$ & trivial \\
    \hfil$\mathbf{IS}$& & inverse semigroups \\
    \midrule
    \hfil$\vG$& $x\ov x = y \ov y$ & groups \\
    \hfil$\vCl$& $x \bar x = \bar x x$ & Clifford semigroups \\
    \hfil$\vSI$& $\bar y x y \, \bar y \bar x y = \bar y \bar x y \,\bar y x y$ & strict inverse semigroups \\
    \midrule
    \hfil$\vSl$& $x^2 = x$ & semilattices \\
    \hfil$\vBS$& $(\bar y x y)^2 = \bar y x y$ & generated by $B_2$ \\
    \hfil$\vBM$& not finitely based~\cite{Kleiman79} & generated by $B_2^1$ \\
    \hfil$\mathbf{A}$& $x^{\omega+1} = x^\omega$ & aperiodic \\
    \bottomrule
  \end{tabular}
\end{minipage}
\caption{Varieties of finite inverse semigroups, their relationships, and defining identities.}
\label{fig:varieties}
\end{figure}

Varieties of finite inverse semigroups form a complete lattice under inclusion. (They are closed under arbitrary intersections and the variety of all finite inverse semigroups $\mathbf{IS}$ serves as a largest element.)
The join of two varieties of finite inverse semigroups $\vU$ and $\vV$, denoted by $\vU \vee \vV$, is the smallest variety of finite inverse semigroups containing both $\vU$ and $\vV$.

The chain $\vT \sse \vSl \sse \vBS \sse \vBM$ forms the bottom of the lattice of the varieties of finite combinatorial (i.e., aperiodic) inverse semigroups.
Herein, $\vT = \vId{x=y} = \vId{x=1}$ denotes the variety of trivial inverse semigroups, $\vSl = \vId{x^2 = x}$ denotes the variety of all finite semilattices, and $\vBS = \vId{(\bar y x y)^2 = \bar y x y}$ and $\vBM$ denote the varieties of finite inverse semigroups generated by the (combinatorial) Brandt semigroup $B_2$ and monoid $B_2^1$, respectively.

\begin{lemma}[{Kleiman~\cite[Lemma~4]{Kleiman77}}]\label{lem:b2_implies_bn}
	For every finite set $\Omega$, $B(\Omega) \in \vBS$ and $B^1(\Omega) \in \vBM$.
\end{lemma}
\begin{proof}
  Let $S = (B_2)^{\Omega}$ be the $\Omega$-fold Cartesian power of the Brandt semigroup $B_2$.
  The set $I = \{ s \in S \mid \pi_x(s) = 0\text{ for some }x \in \Omega\} \leq S$ is an ideal where $\pi_x \colon S \to B_2$ denotes projection onto the $x$ coordinate.
  The Rees quotient $S / I \in \vBS$ is isomorphic to $B(\Omega)$, and the Rees quotient $S^1 / I \in \vBM$ of $S^1 = S \cup \{ 1 \} \leq (B^1_2)^\Omega$ is isomorphic to $B^1(\Omega)$.
\end{proof}

The bottom of the lattice of varieties of finite inverse semigroups is structured as follows. 
These results were originally proved for arbitrary varieties allowing also infinite direct products and then, by \cite{HallJohnston89}, transferred to pseudovarieties (\ie varieties  of finite inverse semigroups in our sense).

\begin{proposition}[{Djadchenko~\cite{Djadchenko77}, Kleiman~\cite{Kleiman76,Kleiman77}; see also \cite{HallJohnston89}}]\label{pro:varieties}
  Let $\vV$ be a variety of finite inverse semigroups.
  Then $\vV$ is subject to each of the following alternatives.
  \begin{enumerate}
    \item Either $\vSl \sse \vV$ or $\vV \sse \vG = \vG \vee \vT$.
    \item Either $\vBS \sse \vV$ or $\vV \sse \vCl = \vG \vee \vSl$.
    \item Either $\vBM \sse \vV$ or $\vV \sse \vSI = \vG \vee \vBS$.
  \end{enumerate}\smallskip
  \hspace{\parindent}Moreover, the intervals $[\vSl, \vCl]$ (or $[\vBS, \vSI]$) and $[\vT, \vG]$ in the lattice of all varieties of finite inverse semigroups are isomorphic via $\vV \mapsto \vV \cap \vG$ and $\vH \mapsto \vH \vee \vSl$ (or $\vH \mapsto \vH \vee \vBS$).
\end{proposition}

\subsection{Algorithmic Problems}
\label{sec:prelims-problems}

The main focus of this work lies on analyzing the algorithmic complexity of two important decision problems for inverse semigroups, and several variants thereof.
The first of these problems is the \emph{membership problem} (\dMemb{}); it is defined as follows.

\begin{decproblem}
  \iitem An inverse semigroup $S$, a subset $\Sigma \sse S$, and an element $t \in S$.
  \qitem Is $t \in U$ where $U = \gen{\Sigma}$?
\end{decproblem}

\noindent
The second main decision problem is the \emph{conjugacy problem} (\dConj{}).

\begin{decproblem}
	\iitem An inverse semigroup $S$, a subset $\Sigma \sse S$, and elements $s,t \in S$.
	\qitem Is $s \sim_U t$ where $U = \gen{\Sigma}$?
\end{decproblem}

Recall that we consider the relative variant of conjugacy (see \cref{def:conjugacy}) meaning that we restrict the conjugating element to the inverse subsemigroup $U$, whereas the elements $s$ and $t$ are not restricted. 
Be aware that this might differ from other references in the literature, especially those concerned with infinite structures. 
Also, unlike what is common for infinite semigroups, we require that $s$ and $t$ are explicitly given as elements of $S$ and not as words over some set of generators.

\mysubparagraph{Input Models.}
As with any decision problem, its algorithmic complexity depends substantially on how its input is provided.
We restrict our attention to two input models, the \emph{Cayley table model} ($\mathbf{CT}$) and the \emph{partial bijection model} ($\mathbf{PB}$).
In the former model, the surrounding inverse semigroup is provided as a complete multiplication table, the so-called \emph{Cayley table} of $S$, and all elements of $S$ (in particular, $\Sigma$ and $t$) are encoded as indices into this table.\footnote{More precisely, a Cayley table of an $n$-element semigroup $S$ is encoded as an array of size $n^2$ each of which entries is encoded using $\ceil{\log n}$ bits and where at position $i + jn$ we find the index of the element obtained by multiplying elements $i$ and $j$ (indices starting at $0$).}

In the partial bijection model, the surrounding inverse semigroup is the symmetric inverse semigroup $\ISym_n$ on $n$ elements with only $n$ provided as part of the input.
All elements of $S = \ISym_n$ (in particular, $\Sigma$ and $t$) are given as partial bijections on $\Omega = \{ 1, \dotsc, n\}$.
More specifically, we assume that each partial bijection is encoded as a complete, ordered list of its images (using a special symbol $\bot$ to denote undefined images)\footnote{Another different representation for permutations, which is also commonly used, is the cycle representation. While two permutations in our encoding (as a partial functions) can be multiplied in \ACz, in the cycle notation, multiplication is $\mathsf{FL}$-complete \cite{CookM87}.}.
For example, $(2,\bot,1)$ encodes the partial bijection on $\{1,2,3\}$ with $1 \mapsto 2$, $3 \mapsto 1$, and undefined on $2$.

We denote by \dMemb[CT]{} and \dMemb[PB]{}, and by \dConj[CT]{} and \dConj[PB]{} the membership and conjugacy problems in the respective input model.
Intuitively, membership and conjugacy are easier to decide in the Cayley table model than in the partial bijection model. 
The fact that the Preston-Wagner representation \cite{Preston54,Wagner52} of an inverse semigroup is efficiently computable allows us to make this intuition precise.

\begin{lemma}\label{lem:CT_to_PB}
  On input of an inverse semigroup $S$ given as a Cayley table, one can compute an embedding $S \to \ISym(S)$ in \ACz.
  Hence, $\dMemb[CT]{} \leq^{\ACz}_m \dMemb[PB]{}$ and $\dConj[CT]{} \leq^{\ACz}_m \dConj[PB]{}$.
\end{lemma}

\begin{proof}
  Every inverse semigroup $S$ acts on itself via multiplication on the right.
  We can restrict this action to obtain a representation via partial bijections.
  Indeed, given $s \in S$, we define the partial map $\rho_s \colon S \to S$ via $t \rho_s = ts$ if $t s\bar s = t$ and $t \rho_s = \bot$ otherwise.
  The resulting map $\rho \colon S \to \ISym(S) \colon s \mapsto \rho_s$ is the desired embedding.
  Now note that encoding of the partial bijection $\rho_s$ is simply the corresponding column of the Cayley table for $S$ with some of its entries replaced by $\bot$.
  It thus remains to argue that we can decide the condition for such a replacement (i.e., whether $t s\bar s \neq t$) in \ACz.
  To see that this is indeed the case, note that we can compute the product of two elements of $S$ in \ACz and, given $s \in S$, we can compute $\bar s$ in \ACz (for $\bar s$ is the unique element of $S$ with $s\bar ss = s$ and $\bar ss\bar s = \bar s$).
\end{proof}

\mysubparagraph{Idempotent Membership and Conjugacy.}
In order to obtain a detailed analysis of the algorithmic complexity, we impose certain restrictions on the allowed inputs.
On the one hand, we consider the \emph{idempotent membership} and \emph{idempotent conjugacy} problems where we require that $s,t \in E(S)$.
We denote these problem variants by \dEMemb[IM]{} and \dEConj[IM]{} where $\mathbf{IM} \in \{\mathbf{CT}, \mathbf{PB}\}$.
The latter, in particular, is closely tied to many other important problems regarding partial symmetries (\eg the set transporter problem; see \cref{sub:groups-pbm}).

\mysubparagraph{Restriction to Varieties.}
The other kind of restriction we impose is on the structure of the inverse subsemigroup $U = \gen{\Sigma}$ under consideration.
More specifically, we consider the above problems with $U$ confined to some fixed class $\vV$ of finite inverse semigroups.
We call these the (idempotent) \emph{membership} and \emph{conjugacy problem for $\vV$}, and denote them by \dMemb[IM]{\vV} and \dConj[IM]{\vV} where $\mathbf{IM} \in \{\mathbf{CT},\mathbf{PB}\}$, respectively.
Throughout, the class $\vV$ will be some variety of finite inverse semigroups such as, e.g., the variety $\vG$ of finite groups.
Be aware that only $U$ is restricted to the class $\vV$, while $S$ and the elements $s,t \in S$ can be arbitrary!

We also consider a more restricted variant of the problems, which we denote with a $\sharp$ superscript (\eg $\dMembS[PB]{}$ and $\dConjS[PB]{}$). 
For these we require in the Cayley table model that $S \in \vV$ and in the partial bijection model that there is some $S\leq \ISym(\Omega)$ with $S \in \vV$ such that $\Sigma\sse S$ and $t \in S$ (resp.\ $s,t \in S$).
For the conjugacy problem we also require that $s \sim_S t$.
We use these restricted variants to show stronger statements for our hardness results.

\subsection{Complexity}\label{sec:complexity}

We assume that the reader is familiar with standard complexity classes such as \PSPACE or \NP; see any standard textbook \cite{AroBar09,pap94} on complexity theory. In particular, if $\cC$ and $\cD$ are complexity classes, then we use the notation $\cC^\cD$ for the class of problems that can be solved in $\cC$ with oracles for a finite set of problems from $\cD$.

\mysubparagraph{Circuit Classes and Reductions.}
The circuit class $\ACz$ is defined as the class of problems decidable by polynomial-size, constant-depth Boolean circuits (where all gates may have arbitrary fan-in).
Likewise \ACz-computable functions are defined.
 We say that a problem $K\sse\{0,1\}^\ast$ is \emph{\ACz-(many-one-)reducible} to $L\sse\{0,1\}^\ast$ if there is an \ACz-computable function $f\colon\{0,1\}^\ast \to \{0,1\}^\ast$ such that $w\in K  \iff f(w) \in L$.
Throughout, we consider only uniform classes meaning that the circuits can be constructed (or verified) efficiently, for details see \cite{Vollmer99}.
 The classes $\qACz$ and $\NC$ are defined analogously to \ACz but allowing circuits of quasipolynomial (\ie $2^{\log^{\bigO(1)}n}$) size (resp.\ polynomial size and depth $\log^{\bigO(1)}n$). 
 
\mysubparagraph{Logarithmic Space.}
We write \LOGSPACE to denote logarithmic space. For many-one reductions computable in logarithmic space we write \LOGSPACE-reductions.
Recall that the composition of two \LOGSPACE-reductions is again a \LOGSPACE-reduction, that the class \LOGSPACE is closed under \LOGSPACE-reductions, and that the class \LOGSPACE is low for itself, \ie $\LOGSPACE^\LOGSPACE = \LOGSPACE$ (see \eg \cite[Lemma 4.17]{AroBar09}).
When talking about \LOGSPACE-hard problems, throughout we refer to \ACz many-one reductions.

Let \dUGAP denote the \emph{undirected graph accessibility problem}, \ie the input is an undirected graph and two vertices $s$ and $t$ and the question is to decide whether  $s$ and $t$ lie in the same connected component.
Note that \dUGAP is \LOGSPACE-hard under \ACz reductions \cite{CookM87} (even if the graphs are restricted to trees).
 The class $ \LOGSPACE^{\dUGAP}$ is also denoted as \SL. 
All our results on \LOGSPACE will rely crucially on the following seminal result, which shows that $\SL = \LOGSPACE$.
\begin{theorem}[Reingold \cite{Reingold08}]
	The problem $\dUGAP$ is in $\LOGSPACE$.
\end{theorem}

\begin{remark}\label{rem:compute-path}
  Using an \dUGAP oracle, we can compute a path between any two vertices of an undirected graph in \LOGSPACE; in fact, we can even compute the vertices of a connected component, or a spanning tree of a connected component in \LOGSPACE, for details see \cite[Lemma 2.4]{NisanT95}.
\end{remark}

\mysubparagraph{Sublinear Time Classes.}

For sublinear time classes, we use random access Turing machines meaning that the Turing machine has a separate address tape and a query state; whenever the Turing machine goes into the query state and on the address tape the number $i$ is written in binary, the $i$-th symbol of the input is read (the content of the address tape is \emph{not} deleted after that).
Apart from that, random access Turing machines work like regular Turing machines.
The class \NPOLYLOGTIME consists of the problems decidable by non-deterministic random access Turing machines in time $\log^{\bigO(1)}n$.

\mysubparagraph{Overview.}

 For an overview over the complexity classes we use in this paper and their relationships, see \cref{fig:complexityclasses}.
We also note that, by \cite{Ruzzo81} and the space hierarchy theorem \cite{StearnsHL65}, we know that $\NC \subsetneqq \PSPACE$.
Moreover, by \cite{CollinsGLW24}, \NPOLYLOGTIME can be simulated by circuits of quasipolynomial size and depth two,\footnote{Thus, in terms of circuit depth, our corresponding results are optimal.} \ie $\NPOLYLOGTIME\sse \qACz$. 
Clearly, \ACz and \LOGSPACE are not contained in \NPOLYLOGTIME (as for example the conjunction of all input bits cannot be computed in \NPOLYLOGTIME).
Even more, by \cite{FurstSS84,Hastad86}, \qACz does not contain any \LOGSPACE-hard problem as it cannot compute, for instance, \prob{parity}.

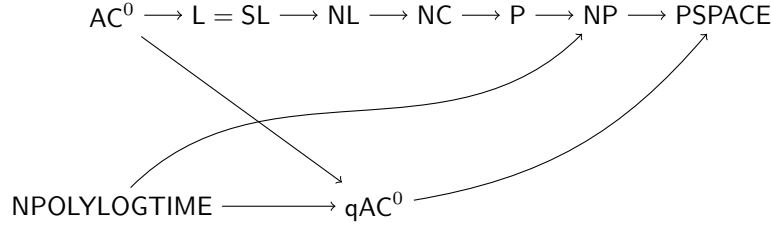
\begin{figure}
  \centering
	\begin{tikzpicture}[node distance=.5cm, auto]
		
		% Nodes (complexity classes)
		\node (PSPACE) {\(\PSPACE\)};
		\node (NP) 		[left=of PSPACE] {\(\NP\)};
		\node (PTIME) 	[left=of NP] {\(\PTIME\)};
		\node (NC) 		[left=of PTIME] {\(\NC\)};
		\node (NL) 		[left=of NC] {\(\NL\)};
		\node (SL) 		[left=of NL] {\(\LOGSPACE = \SL\)};
		\node (ACz) 	[left=of SL] {\(\ACz\)};
		\node (NPOLYLOGTIME) [below=2cm of ACz] {\(\NPOLYLOGTIME\)};
		\node (qACz) 	[right=1.5cm of NPOLYLOGTIME] {\(\qACz\)};
		
		% Arrows (subset relations, bottom to top)
		\draw[->] (NPOLYLOGTIME) -- (qACz);
    \draw[->] (qACz) to[bend right=20] (PSPACE);
    \draw[->] (NPOLYLOGTIME) to[out=45, in=-135](NP);
		\draw[->] (NP) -- (PSPACE);
		\draw[->] (ACz) -- (qACz);
		\draw[->] (ACz) -- (SL);
		\draw[->] (SL) -- (NL);
		\draw[->] (NL) -- (NC);
		\draw[->] (NC) -- (PTIME);
		\draw[->] (PTIME) -- (NP);
		
	\end{tikzpicture}
  \caption{Complexity classes in this paper (all circuit classes are assumed to be uniform).}\label{fig:complexityclasses} 
\end{figure}

\subsection{Straight-Line Programs}

Let $S$ be an inverse semigroup and $\Sigma \sse S$.
A \emph{straight-line program} (SLP) over $\Sigma$ is a finite sequence $(s_1, \dots, s_k) \in S^k$ such that for all $i$ either $s_i \in \Sigma$, or $s_i = s_js_\ell$ for some $j,\ell < i$, or $s_i = \ov{s_j}$ for some $j < i$. 
An SLP as above \emph{computes} an element $s \in S$ if $s \in \{s_1, \dots, s_k\}$.

Note that for $\Sigma \sse S$ closed under formation of inverses (as is the case for generating sets by our convention) the rule allowing for $s_i = \ov{s_j}$ could have been omitted from the definition of a straight-line program over $\Sigma$.
Note also that our definition of SLPs is according to Babai and Szemer\'edi \cite{BabaiS84} -- other authors define them slightly differently via circuits (or equivalently context-free grammars). The difference is that in our definition the evaluation of the SLP in the semigroup $S$ is already part of its definition.

\begin{definition}\label{def:polylog-slps}
  We say that a class $\cC$ of finite inverse semigroups \emph{admits polylogarithmic SLPs} if there exists a polynomial $P$ such that, for all $S \in \cC$ and all generating sets $\Sigma \subseteq S$, every element $s \in S$ is computed by some SLP over $\Sigma$ of length at most $P(\log \abs{S})$.
\end{definition}

The analogous property for semigroups and monoids was studied \ifAnonimous in~\cite{Fleischer22,Fleischer19diss}
\else by the first author~\cite{Fleischer22,Fleischer19diss}
\fi
 under the name \emph{polylogarithmic circuits property}. 
\ifAnonimous
A straight-forward guess-and-check approach yields the following result.
\else Using a straight-forward guess-and-check approach, we obtain the following result~-- for a proof see \cite[Corollary 5.2]{Fleischer19diss}.
\fi

\begin{lemma}[Fleischer {\cite[Corollary 5.2]{Fleischer19diss}}]\label{lem:SLP-npolylogtime}
  Let $\cC$ be a class of finite (inverse) semigroups admitting polylogarithmic SLPs.
  Then the problem $\dMemb[CT]{\cC}$ is in $\NPOLYLOGTIME$.
\end{lemma}

\begin{remark}
	Notice that in \cite{CollinsGLW24arxiv}	membership in quasigroups in the Cayley table model has been shown to be decidable in the class $\exists^{\log^2}\mathsf{DTISP}(\mathrm{polylog}, \mathrm{log})$, meaning that, after non-deterministically guessing $\bigO(\log^2 n)$ bits, it can be verified  deterministically in time $\log^{\bigO(1)} n$ with space restricted to $\bigO(\log n)$. 
	Our results could be strengthened to the similar (yet slightly larger) class $\exists^{\log^{k+1}}\mathsf{DTISP}(\mathrm{polylog}, \mathrm{log})$ where $k$ is the degree of a polynomial $P$ as in \cref{def:polylog-slps}.
  Note that $k \leq 2$ for the variety $\vCl$ of finite Clifford semigroups by \cref{lem:clifford-slp}, matching the bound for the variety $\vG$ of finite groups obtained by Babai and Szemer\'edi~\cite{BabaiS84}; see \cref{lem:groups-slp}.
	Nevertheless, as $\NPOLYLOGTIME = \exists^{\log^{\bigO(1)}}\mathsf{DTISP}(\mathrm{polylog}, \mathrm{log})$, this would give little additional insight and we refrain from doing so for the sake of a cleaner presentation.
\end{remark}

\section{Membership and Conjugacy in Groups}\label{sec:groups}

Groups are a primary example for inverse semigroups.
Therefore, let us start exploring some known result and new observations about the membership and conjugacy problems in groups.

The following characterization of the variety of finite groups is well known (see, \eg \cite{HallJohnston89}).

\begin{lemma}\label{lem:characterization-groups}
  Let $S$ be a finite inverse semigroup.
  Then the following are equivalent.
  \begin{enumerate}
    \item The inverse semigroup $S$ is contained in $\vG = \vId{x \bar x = 1}$.
    \item The two-element semilattice $Y_2$ does not divide $S$.
  \end{enumerate}
\end{lemma}

\subsection{The Cayley Table Model}

In the Cayley table model, deciding the membership and conjugacy problems for groups is comparatively easy.
The best currently known approach \cite{CollinsGLW24arxiv} is based on the non-deterministic computation of a succinct representation of a target element as a product of generators.

We pursue a similar idea here and use SLPs for succinct representation.
In the case of groups, this approach is afforded by the following Reachability Lemma.

\begin{lemma}[Babai, Szemer\'edi~{\cite[Theorem~3.1]{BabaiS84}}]\label{lem:groups-slp}
	The variety $\vG$ of finite groups admits polylogarithmic SLPs.
	More precisely, for every group $G$ and generating set $\Sigma \subseteq G$, every element of $G$ is computed by an SLP over $\Sigma$ of length $\mathcal{O}(\log^2 \abs{G})$.
\end{lemma}
For the first part of the following result, see \cite{Fleischer19diss}.
\begin{proposition}\label{pro:ctm-groups}
  The problems \dMemb[CT]{\vG} and \dConj[CT]{\vG} are in \NPOLYLOGTIME{}.
\end{proposition}

\begin{proof}
  The combination of \cref{lem:SLP-npolylogtime} and \cref{lem:groups-slp} shows that \dMemb[CT]{\vG} is decidable in \NPOLYLOGTIME{}.
  To see that this is also true for \dConj[CT]{\vG}, we observe that one can simply guess a conjugating element, thereby reducing the problem to \dMemb[CT]{\vG}.
\end{proof}

\subsection{The Partial Bijection Model}\label{sub:groups-pbm}

In the partial bijection model, the groups under consideration are permutation groups and it is in this latter setting that the membership and conjugacy problems have been widely studied.
However, the problems \dMemb[PB]{\vG} and \dConj[PB]{\vG} are also subtly different from the corresponding problems for permutation groups simply because the former allow for more possible inputs (partial bijections instead of bijections). 
In case of the membership problem this distinction is mostly artificial (and can be  resolved by appropriate \ACz reductions).

\begin{proposition}[Babai, Luks, and Seress \cite{BabaiLS87}]\label{pro:pbm-groups-membership}
  The problem \dMemb[PB]{\vG} is in \NC{}.
\end{proposition}

\begin{proof}
	Let $\Sigma \sse \ISym(\Omega)$ such that $U = \gen{\Sigma}$ is a group and $t \in \ISym(\Omega)$ denote our input.
	First, observe that $\dom(u) = \dom(v)$ for all $u,v \in U$ as $U$ is a group. 
	Hence, to check membership, we first check whether $\dom(t) = \ran(t) = \dom(u)$ for some $u \in \Sigma$. 
  If this is not the case, then $t \not\in U$. Otherwise, we use the algorithm for permutation groups \cite{BabaiLS87} to test whether $t \in\gen{\Sigma}$, where we interpret $t$ and all elements of $\Sigma$ as permutations on the set $\Omega^U = \dom(t) \sse \Omega$.
\end{proof}

We complement the above with the following hardness result.

\begin{proposition}\label{lem:nontrivial-group-hard}
	Let $\vV$ be a variety of finite inverse semigroups containing a non-trivial group.
	Then \dMemb[PB]{\vV} as well as its restricted variant $\dMembS[PB]{\vV}$ are \LOGSPACE-hard.
\end{proposition}
\begin{proof}
	Let us reduce \dUGAP to \dMembS[PB]{\vV}. 
	Let $G \in \vV$ denote a non-trivial group with some non-trivial element $g \in G$.
	We can interpret $G$ as a permutation group acting on itself.
	Given an undirected graph $\Gamma = (V,E)$, set $S = G^V \in \vV$ (which can be interpreted as a permutation group acting on $\bigsqcup_{v\in V} G$). 
  For each $v \in V$ we define $g_v\colon V \to G$ with $g_v(v) = g$ and $g_v(u) = 1$ otherwise, and for each pair $(u,v) \in V\times V$ we define $g_{uv} = \ov g_{u} g_{v}$.
  Clearly, $g_{st} \in \gen{\,g_{uv} \mid \{u,v\} \in E\,}$ if and only if $s \in V$ and $t \in V$ are in the same component of $\Gamma$.
\end{proof}

The complexity of the conjugacy problem for groups is more intricate.
On the one hand, the problem is clearly in \NP{} due to the existence of short SLPs (\cref{lem:groups-slp}).
This observation holds for permutation groups as well as groups in the partial bijection model.

On the other hand, the conjugacy problem is \GI-hard for permutation groups, as was observed by Luks \cite{Luks93}.
Luks also exhibited several other problems for permutation groups that are polynomial-time equivalent to conjugacy, among which is the \emph{set transporter problem}.

\begin{decproblem}
	\iitem A group $G \leq \Sym(\Omega)$ given by generators, and sets $\Delta_s, \Delta_t \sse \Omega$.
	\qitem Does there exist an element $g \in G$ such that $\Delta_s^g = \Delta^{\mathstrut}_t$?
\end{decproblem}

The conjugacy problem for permutation groups is clearly a special case of the conjugacy problem for $\vG$ in the partial bijection model, and so is the set transporter problem.
In fact, the latter is precisely the idempotent conjugacy problem for $\vG$ (recall that the idempotents of the symmetric inverse semigroup $\ISym(\Omega)$ are in canonical bijection with the subsets of $\Omega$).

\begin{lemma}
  The problem \dConj[PB]{\vG} is $\ACz$-reducible to \dEConj[PB]{\vG}.
\end{lemma}

\begin{proof}
	Let $U \leq \ISym(\Omega)$ be a group, and let $s,t \in \ISym(\Omega)$ be partial bijections.
	We may assume that $U$ consists of bijections of $\Omega$, for it consists of bijections on some $\Omega^U \subseteq \Omega$ and if $\mathrm{dom}(s) \cup \mathrm{dom}(t) \cup \mathrm{ran}(s) \cup \mathrm{ran}(t) \not\subseteq \Omega^U$, then $s \not\sim_U t$.
	Now let $\Delta_s, \Delta_t \subseteq \Omega\times\Omega$ be the graphs of $s$ and $t$, respectively.
	We claim that $\Delta^u_s = \Delta^{\mathstrut}_t$ for some $u \in U$ with respect to the diagonal action of $U$ on $\Omega \times \Omega$ if and only if $\bar u s u = t$ and $s = u t \bar u$.
	Indeed, a short calculation shows that $(x, x^s) \in \Delta_s$ and $(x, x^s)^u \in \Delta_t$ if and only if $x^s = x^{ut\bar u}$ with both sides defined.
	
	Provided that the input is suitably encoded, transforming an instance of the conjugacy problem to its corresponding instance of the idempotent conjugacy problem, as above, (or to a default instance if $\mathrm{dom}(s) \cup \mathrm{dom}(t) \cup \mathrm{ran}(s) \cup \mathrm{ran}(t) \not\subseteq \Omega^G$) is possible with $\ACz$-circuits.
\end{proof}

The above discussion can be summarized as follows.

\begin{proposition}\label{pro:bpm-groups-conjugacy}
  Both of the problems \dConj[PB]{\vG} and \dEConj[PB]{\vG} are polynomial-time equivalent to the conjugacy problem for permutation groups; hence, \GI-hard and in \NP{}.
\end{proposition}

Notice that the situation is different for the restricted variants \dConjS[PB]{\vG} and \dEConjS[PB]{\vG}.
As conjugacy in the symmetric group $S_n$ can be tested in \L, \dConj[PB]{\vG} can be reduced to \dConjS[PB]{\vG}; thus, the latter is as difficult as the general case.
On the other hand, the restricted problem variant \dEConjS[PB]{\vG} is trivial as every group only contains a single idempotent.

\section{Membership and Conjugacy in Clifford Semigroups}\label{sec:clifford}

We now turn to Clifford semigroups, which constitute the smallest variety of finite inverse semigroups $\vCl$ to properly contain the variety of finite groups $\vG$.
Our goal is to show that the membership and conjugacy problems for Clifford semigroups are essentially as complex as the corresponding problems for groups.
To this end, let us first recall the following well known characterization of Clifford semigroups (see, \eg \cite{Petrich84,Law99}).

\begin{lemma}\label{lem:characterization-clifford}
  Let $S$ be a finite inverse semigroup.
  Then the following are equivalent.
  \begin{enumerate}
    \item The inverse semigroup $S$ is contained in $\vCl = \vId{x\bar x = \bar x x}$.
    \item The Brandt semigroup $B_2$ does not divide $S$.
  \end{enumerate}
\end{lemma}

The following simple observations regarding the structure of the idempotents $E(S)$ of a Clifford semigroup $S$ are also of crucial importance for our discussion.

\begin{lemma}\label{lem:idempotent_nesting-clifford}
	Let $S \in \vCl$.
	If $s_1, \dotsc, s_n \in S$, then 
	\[
	s_1 s_2 \dotsc s_n \bar s_n \dotsc \bar s_2 \bar s_1 = s_1 \bar s_1 s_2 \bar s_2 \dotsc s_n \bar s_n.
	\]
\end{lemma}
\begin{proof}
	The identity holds for all $S \in \vG$, since both sides are idempotent, and for all $S \in \vSl$, since $S \in \vSl$ implies $S = E(S)$ and idempotents in an inverse semigroup commute.
	Hence, the identity also holds for all $S \in \vG \vee \vSl = \vCl$.
\end{proof}

\begin{lemma}\label{lem:idempotent_generators-clifford}
	Let $S \in \vCl$.
	If $S = \langle s_1, \dotsc, s_k \rangle$, then $E(S) = \langle s_1 \bar s_1, \dotsc, s_k \bar s_k \rangle$.
\end{lemma}
\begin{proof}
	Suppose that $S = \langle s_1, \dotsc, s_k \rangle$ and let $e \in E(S)$.
	Then, by \cref{lem:idempotent_nesting-clifford},  
	\[
	e = s_{i_1} \dotsc s_{i_n} = s_{i_1} \dotsc s_{i_n} \bar s_{i_n} \dotsc \bar s_{i_1} = s_{i_1} \bar s_{i_1} \dotsc s_{i_n} \bar s_{i_n} \in \langle s_1 \bar s_1, \dotsc, s_k \bar s_k \rangle.\qedhere
	\]
\end{proof}

\subsection{The Cayley Table Model}\label{sec:Clifford-CT}

As is the case for groups, Clifford semigroups afford succinct representations of their elements. 

\begin{lemma}\label{lem:semilattice-slp}
	Every finite semilattice $E$ admits SLPs of length $\mathcal{O}(\log {}\lvert E \rvert)$.
\end{lemma}
\begin{proof}
	Let $E$ be a finite semilattice generated by $\Sigma \subseteq E$.
	Given $e \in E$, let $e_1, \dotsc, e_n \in \Sigma$ such that $e = e_1 \dotsc e_n$ and $n$ is minimal with this property.
	The elements $e_I = \prod_{i \in I} e_i \in E^1$ with $I \subseteq \{1, \dotsc, n\}$ are pairwise distinct.
	Indeed, if $e_I = e_J$ for some $I \neq J$ and $i \in I \setminus J$, say, then $e = e_1 \dotsc e_{i-1} e_{i+1} \dotsc e_n$ which contradicts the minimality of $n$.
	Hence, $\lvert E \rvert \geq 2^n - 1$.
\end{proof}
\ifAnonimous For the following result, see also \cite[Lemma 4.10]{Fleischer19diss}.
\else The following result is also part of the first author's dissertation \cite[Lemma 4.10]{Fleischer19diss}.
\fi
\begin{lemma}\label{lem:clifford-slp}
	Every finite Clifford semigroup $S$ admits SLPs of length $\mathcal{O}(\log^2 \abs{S})$.
\end{lemma}
\begin{proof}
	Let $S \in \vCl$ be generated by $\Sigma \subseteq S$ and let $t \in S$.
	The element $t \bar t \in E(S)$ can be computed by an SLP of length $\mathcal{O}(\log {} \lvert E(S)\rvert)$ over $\{ s \bar s \mid s \in \Sigma \} \subseteq E(S)$ by \cref{lem:semilattice-slp}.
	Moreover, $t$ is contained in the subgroup $S' = \{ s \in S \mid s \bar s = t \bar t \, \} \leq S$, which is generated by the set $\Sigma' = \{ t \bar t s \in S \mid s \in \Sigma, s \bar s \geq t \bar t\, \} \subseteq S'$.
	Therefore, $t$ can be computed by an SLP of length $\mathcal{O}(\log^2 \lvert S' \rvert)$ over $\Sigma'$ by \cref{lem:groups-slp}.
	Combining these two observations, we see that $t$ can be computed by an SLP of length $\mathcal{O}(\log {} \lvert E(S) \rvert + \log^2 {}\lvert S' \rvert) \subseteq \mathcal{O}(\log^2 \lvert S \rvert)$ over $\Sigma$.
\end{proof}

The following is an immediate consequence (for membership see also \cite{Fleischer19diss}).

\begin{proposition}\label{pro:ctm-clifford}
  The problems \dMemb[CT]{\vCl} and \dConj[CT]{\vCl} are in \NPOLYLOGTIME{}.
\end{proposition}

\subsection{The Partial Bijection Model}\label{sub:clifford-pb}

Even in the partial bijection model, the complexity of the membership and conjugacy problems for Clifford semigroups is essentially equivalent to those for groups.

\begin{lemma}\label{lem:clifford-idempotent_cover}
  Let $U \in \vCl$ be generated by $\Sigma \sse \ISym(\Omega)$.
	On input $e \in E(\ISym(\Omega))$ and $\Sigma$, the minimal idempotent $\hat e \in E(U) \cup \{ 1 \}$ such that $\hat e \geq e$ can be computed in $\ACz$.
\end{lemma}
\begin{proof}
  This follows from the fact that $\hat e = \prod \{ u \bar u \mid u \in \Sigma \cup \{1\}, u \bar u \geq e \}$ by \cref{lem:idempotent_generators-clifford}.
  Moreover, note that $u \bar u$ is the idempotent associated with the set $\dom(u) \sse \Omega$ and, as such, the condition $u \bar u \geq e$ is equivalent to $\dom(u) \supseteq \dom(e)$ which can be verified in \ACz.
  The product in question and, in fact, the product $\prod_{f \in F} f$ of any set of idempotents $F \sse E(\ISym(\Omega))$ can also be computed in \ACz; it is the idempotent associated with the set $\bigcap_{f \in F} \dom(f)$.
\end{proof}

\begin{proposition}\label{pro:pbm-clifford}
  If $\vH \sse \vG$ is a variety of finite groups, then the problems \dMemb[PB]{\vH\vee\vSl} and \dConj[PB]{\vH\vee\vSl} are $\ACz$-reducible to \dMemb[PB]{\vH} and \dConj[PB]{\vH}, respectively.
\end{proposition}

\begin{proof}
	Let $U \leq \ISym(\Omega)$ with $U \in \vH \vee \vSl$ generated by $\Sigma \subseteq U$.
	In the case of membership, we are given an additional element $t \in \ISym(\Omega)$.
  Using \cref{lem:clifford-idempotent_cover}, we compute the minimal idempotent $\hat e \in E(U) \cup \{1\}$ with $\hat e \geq t\bar t$ and (as part of this computation) verify that $\hat e \in E(U)$.
  Next, we compute the generating set $\Sigma_{\hat e} = \{ \hat e u \mid u \in \Sigma, \hat e \leq u \bar u \}$ of the $\gH$-class $U_{\hat e} = \{ u \in U \mid u \bar u = \hat e \} \leq U$ (which is a group).
	Then $t \in U$ if and only if $t \in U_{\hat e}$.

	In case of conjugacy, on input $s,t \in \ISym(\Omega)$, we perform the reduction with $\hat e \in E(U) \cup \{1\}$ minimal such that $\hat e \geq s\bar s \vee \bar s s \vee t \bar t \vee \bar t t$ and the $\gH$-class $U_{\hat e} \leq U^1$ of $\hat e$.
	Note that $\bar u s u = t$ and $u t \bar u = s$ with $u \in U^1$ imply $u \bar u \geq \hat e$ and, therefore, $\bar u' s u' = t$ and $u' t \bar u' = s$ with $u' = \hat e u \hat e \in U_{\hat e}$.
	As such, $s \sim_U t$ if and only if $s \sim_{U_{\hat e}} t$.
\end{proof}

\begin{corollary}\label{cor:pbm-clifford-groups}
  The problems \dMemb[PB]{\vCl} and \dConj[PB]{\vCl} are in \NC and \NP, respectively.
\end{corollary}
\begin{corollary}[Beaudry, McKenzie, Thérien \cite{BeaudryMT92}]\label{cor:pbm-semilattices}
  The problem \dMemb[PB]{\vSl} is in \ACz.
\end{corollary}

By \cref{lem:CT_to_PB}, it follows from \cref{cor:pbm-semilattices} that also \dMemb[CT]{\vSl} is in \ACz. On the other hand, the following important question remains open.

\begin{question}
  Are the problems \dMemb[CT]{\vG} and  \dMemb[CT]{\vCl} in \ACz?
\end{question}

\section{Membership and Conjugacy in Strict Inverse Semigroups}

In this section we show that, in the partial bijection model, the membership and conjugacy problem for $\vSI$ are \LOGSPACE-reducible to the corresponding problems for $\vG$.
In the Cayley table model these problems are \LOGSPACE-complete for $\vSI$ as we will show in \cref{sec:SL-Cayley}.

Our reduction is explicit and based on special properties of the representation theory of the variety $\vSI$.
For now, let us recall the following characterization of $\vSI$ (see, e.g., \cite{HallJohnston89}).

\begin{lemma}\label{lem:gb-characterization}
  Let $S$ be a finite inverse semigroup.
  Then the following are equivalent.
  \begin{enumerate}
    \item The inverse semigroup $S$ is contained in $\vSI$.
    \item If $e,f_1,f_2 \in E(S)$ with $e \geq f_1, f_2$ and $f_1 \gJ f_2$, then $f_1 = f_2$.
    \item The Brandt monoid $B^1_2$ does not divide $S$.
  \end{enumerate}
\end{lemma}

We will also need the following analogue of \cref{lem:idempotent_nesting-clifford}.

\begin{lemma}\label{lem:gb-idempotent_nesting}
	Let $S \in \vSI$.
	If $s_1, \dotsc, s_n \in S$ with $s_i \bar s_i = \bar s_i s_i$ for all $1 \leq i \leq n$ then 
	\[
	s_1 s_2 \dotsc s_n \bar s_n \dotsc \bar s_2 \bar s_1 = s_1 \bar s_1 s_2 \bar s_2 \dotsc s_n \bar s_n.
	\]
\end{lemma}
\begin{proof}
	Note that $s \bar s = \bar s s$ if and only if $s = \bar x x x$ for some $x \in S$ (e.g., $x=s$).
	As such, the implication can be written as an identity.
	This identity holds for all $S \in \vG$, since both sides are idempotent, and for $B_2$, since $s \in B_2$ with $s \bar s = \bar s s$ implies $s \in E(B_2)$ and idempotents in an inverse semigroup commute.
	Hence, the identity also holds for all $S \in \vG \vee \vBS = \vSI$.
\end{proof}

\subsection{Representations of Strict Inverse Semigroups}\label{sub:munn}

Our goal is to develop an efficient description of the local structure of a strict inverse semigroup $U \leq \ISym(\Omega)$ based on its action on $\Omega$ and on $U$-invariant subsets $\Delta \sse \Omega$.
To this end, we will show how to obtain a generating set for each $\gD$-class of $U$ (as a groupoid) from a given generating set $\Sigma \sse U$.
As shown in the next subsection, this can be used to reduce the membership and conjugacy problems for $U$ to an appropriate $\gD$-class and, ultimately, to a single $\gH$-class, i.e., to a subgroup of $U$.
The attentive reader might notice that this strategy also underlies our approach for Clifford semigroups in \cref{sec:clifford} (and even the transition from groups of partial bijections to permutation groups in \cref{sec:groups}).

Given an inverse semigroup $U \leq \ISym(\Omega)$, we say that a set $\Delta \sse \Omega$ is \emph{$U$-invariant} if $\Delta^s \sse \Delta$ for all $s \in U$.
Equivalently, the set $\Delta$ is $U$-invariant if the idempotent $e_\Delta \in E(\ISym(\Omega))$ associated with it centralizes $U$, i.e., $s e_\Delta = e_\Delta s$ holds for all $s \in U$.
Clearly, each $\Delta \sse \Omega$ generates a $U$-invariant subset $\Delta^U \coloneqq \bigcup_{s \in U} \Delta^s \sse \Omega$.
If each point $x \in \Delta$ is in the domain of some $s \in U$, then $\Delta^U$ is the minimal $U$-invariant subset containing $\Delta$.
Conversely, the set of all points $x \in \Omega$ that are contained in the domain of some $s \in U$ is precisely $\Omega^U$.

The following lemma describes the structure of the \emph{orbit} $x^U \coloneqq \{x\}^U$ of a point $x \in \Omega$ under an inverse semigroup $U \in \vSI$, i.e., of a minimal non-empty $U$-invariant subset.
Throughout, it will be helpful to keep in mind that $y \in x^U$ if and only if $x \in y^U$.

\begin{lemma}\label{lem:gb-orbits}
  Let $U \leq \ISym(\Omega)$ with $U \in \vSI$, $\Delta = x^U$ for some $x \in \Omega$ and $s_1, s_2 \in U$.
  Then $\dom(s_1) \cap \Delta = \dom(s_2) \cap \Delta$ or $\dom(s_1) \cap \dom(s_2) \cap \Delta = \emptyset$.
  In other words, every orbit $x^U$ is partitioned by the non-empty sets of the form $\dom(s) \cap x^U$ with $s \in U$.
\end{lemma}

\begin{proof}
  Suppose to the contrary that there exist $x_1, x_2 \in x^U$ with $x_1 \in \dom(s_1) \cap \dom(s_2)$ and $x_2 \in \dom(s_1) \setminus \dom(s_2)$.
  Since $x_1, x_2 \in x^U$, there exists some $t \in U$ with $x_1^t = x_2$.
  Observe that $e = s_1 \bar s_1 \in U$ and $s = s_2 \bar s_2 t \in U$ satisfy  $x_1^e = x_1$, $x_2^e = x_2$, $x_1^s = x_2$, $x_2^s = \bot$.
  Hence, the restrictions of $e$ and $s$ to $\Omega' = \{x_1,x_2\} \sse \Omega$ generate the Brandt monoid $B^1(\Omega') \leq \ISym(\Omega')$.
  As such, $B_2^1$ divides $U$ which contradicts $U \in \vSI$.
\end{proof}

Important to our cause are the elements of $U \leq \ISym(\Omega)$ that act on all orbits contained in some $U$-invariant set $\Delta \sse \Omega$.
Formally, we say that $s \in U$ is \emph{$\Delta$-large} provided that $(\dom(s)\cap \Delta)^U = \Delta$ or, equivalently, $(\ran(s)\cap \Delta)^U = \Delta$.
We claim that if $s \in U$ is $\Delta$-large and $t \in U$ satisfies $s \gJle t$, then $t$ is $\Delta$-large itself.
Indeed, if $s$ is $\Delta$-large and if $s \gRle t$ or $s \gLle t$, then $t$ is $\Delta$-large since $\dom(s) \sse \dom(t)$ or $\ran(s) \sse \ran(t)$, respectively; finally, if $s \gJle t$, then $s \gRle r \gLle t$ for some $r \in U$ from which we conclude that $t$ is $\Delta$-large.

\begin{lemma}\label{lem:gb-large}
  Let $U \leq \ISym(\Omega)$ with $U \in \vSI$ and $s,t \in U$. 
  Further, let $\Delta \sse \Omega$ be $U$-invariant.
  If $s$ is $\Delta$-large and $e_\Delta s \leq e_\Delta t$, then $t$ is $\Delta$-large and $e_\Delta s = e_\Delta t$.
\end{lemma}

\begin{proof}
  Suppose that $s,t \in U$ are such that $e_\Delta s \leq e_\Delta t$.
  Then \[
    \dom(s) \cap \Delta = \dom(e_\Delta s) \sse \dom(e_\Delta t) = \dom(t) \cap \Delta.
  \]
  In particular, if $s$ is $\Delta$-large, then so is $t$.
  Let us show that in this case the above is an equality, i.e., $\dom(s) \cap \Delta = \dom(t) \cap \Delta$ and thus $e_\Delta s = e_\Delta t$.
  Consider some $x \in \dom(t) \cap \Delta$.
  Then $x \in x^U \sse \Delta = (\dom(s) \cap \Delta)^U$; hence, $\dom(s) \cap x^U \neq \emptyset$.
  As $\dom(s) \cap x^U \sse \dom(t) \cap x^U$, we conclude that $\dom(s) \cap x^U = \dom(t) \cap x^U$ by \cref{lem:gb-orbits}; hence, $x \in \dom(s)$. 
\end{proof}

For the remainder of this section, we restrict our attention to the inverse semigroups contained in the variety of interest.
We assume throughout that $U \in \vSI$ is generated by the set $\Sigma \sse \ISym(\Omega)$, which is closed under formation of inverses.

The construction we use is closely tied to the representation of an inverse semigroup via its action on idempotents by conjugation -- the Munn representation (see \eg \cite{Petrich84,Law99}).
Indeed, the following graph can be obtained as (part of) the Schreier graph of such an action.

\begin{definition}\label{def:gb-munn}
  Let $\Delta \sse \Omega$ be $U$-invariant and let $\Sigma_\Delta \coloneqq \{ u \in \Sigma \mid u \text{ is $\Delta$-large} \}$.
  We then define the graph $\mathrm{M}(\Delta; \Sigma)$, which we call the \emph{Munn graph} at $\Delta$ with respect to $\Sigma$, as follows.
  Its set of vertices is $E_\Delta \coloneqq \{e_\Delta u\bar u \mid u \in \Sigma_\Delta\} \sse E(\ISym(\Omega))$ and its set of edges is $\Sigma_\Delta$, where the edge $u \in \Sigma_\Delta$ connects its source vertex $e_\Delta u\bar u$ to its target vertex $e_\Delta \bar uu$.
\end{definition}

Recall that, as $\Delta$ is $U$-invariant, we have $e_\Delta u\bar u = ue_\Delta \bar u = u\bar u e_\Delta$ for all $u \in \Sigma$.
The Munn graph is undirected in the sense that every edge $u \in \Sigma_\Delta$ has an inverse, viz.\ $\bar u$.
As indicated above, paths in $\mathrm{M}(\Delta; \Sigma)$ encode the action of $U$ on $E_\Delta$ by conjugation.

\begin{lemma}\label{lem:gb-munn-paths}
  Let $u_1, \dotsc, u_n \in \Sigma$ and $e_s,e_t \in E_\Delta$ for some $U$-invariant subset $\Delta \sse \Omega$.
  Then the product $u = u_1 \cdots u_n \in U$ satisfies $\bar u e_s u = e_t$ if and only if the sequence $(u_1, \dotsc, u_n)$ is a path from $e_s$ to $e_t$ in the Munn graph $\mathrm{M}(\Delta; \Sigma)$ (and thus, in particular, $u_1, \dotsc, u_n \in \Sigma_\Delta$).
\end{lemma}

\begin{proof}
  Since the general case follows by a simple induction on the number $n$, we will only consider the case of a single generator $u = u_1 \in \Sigma$.
  If $u$ is an edge from $e_s$ to $e_t$ in $\mathrm{M}(\Delta; \Sigma)$, then $u \in \Sigma_\Delta$ and $\bar u e_s u = \bar u e_\Delta u\bar u u = e_\Delta \bar u u \bar u u = e_\Delta \bar u u = e_t$.
  
  Conversely, let us now assume that $\bar u e_s u = e_t$, and let $u_s, u_t \in \Sigma_\Delta$ with $e_s = e_\Delta u_s \bar u_s$ and $e_t = e_\Delta u_t \bar u_t$.
  Using the fact that $e_\Delta \geq e_s$, we obtain $e_\Delta \bar u u = \bar u e_\Delta u \geq \bar u e_s u = e_t = e_\Delta u_t \bar u_t$.
  By \cref{lem:gb-large}, $\bar u u$ is $\Delta$-large (and thus so is $u\bar u$) and $e_\Delta \bar u u = e_t$.
  Applying \cref{lem:gb-large} again, the inequality $e_\Delta u \bar u = ue_\Delta \bar u u \bar u = u e_t \bar u \leq e_s = e_\Delta u_s\bar u_s$ implies $e_\Delta u\bar u = e_s$.
\end{proof}

Next, we show that every $\gD$-class of $U$ can be recovered from the Munn graph $\mathrm{M}(\Delta; \Sigma)$ at an appropriately chosen $\Delta \sse \Omega$, beginning with the idempotents of such a class.
Recall that the conditions $e \gD f$, $e \gJ f$, and $e \sim f$ (i.e., $e$ and $f$ are conjugate) are equivalent for idempotents $e,f \in E(S)$ of a finite inverse semigroup $S$ (see \cref{lem:idem_conj_is_green_j}).

\begin{lemma}\label{lem:gb-munn-vertices}
  Let $\Delta = (\dom(e))^U$ for some $e \in E(U)$.
  Then the set $\{f \in E(U) \mid e \sim_U f\}$ is the vertex set of a connected component of $\mathrm{M}(\Delta; \Sigma)$.
\end{lemma}

\begin{proof}
  Let $e,f \in E(U)$ with $e \sim_U f$, i.e., with $\bar s e s = f$ and $s f \bar s = e$ for some $s \in U$.
  Then
  \[
    (\dom(e))^U = (\dom(f))^{\bar s U} \sse (\dom(f))^U = (\dom(e))^{s U} \sse (\dom(e))^U
  \]
  which shows that $(\dom(e))^U = (\dom(f))^U$.
  We now prove that $e$ is a vertex of $\mathrm{M}(\Delta; \Sigma)$ which, by the preceding calculation, then also holds for $f$.
  Clearly, $\dom(e) \sse \Delta = (\dom(e))^U$.
  Hence, $e = e_\Delta e$ and $e$ is $\Delta$-large.
  Since $\Sigma \sse U$ is a generating set and $e\in E(U)$, we have $e \leq u\bar u$ for some $u \in \Sigma$.
  Therefore, we have $e = e_\Delta e \leq e_\Delta u\bar u$ which, by \cref{lem:gb-large}, implies that $e = e_\Delta u\bar u$ and $u \in \Sigma_\Delta$.

  The fact that $f$ and $e$ are connected in $\mathrm{M}(\Delta; \Sigma)$ now follows from \cref{lem:gb-munn-paths}, as does the fact that every vertex of $\mathrm{M}(\Delta; \Sigma)$ connected to $e$ is some $f \in E(U)$ with $e \sim_U f$.
\end{proof}

A $\gD$-class $D$ of an inverse semigroup $S$ restricts to a groupoid with objects $D \cap E(S)$ and morphisms $D$ where $s \in D$ is a morphism from $s \bar s$ to $\bar s s$ (see~\cite{CliffordMiller56}).
In the case at hand, we have already identified the objects as the vertices of a connected component of $\mathrm{M}(\Delta; \Sigma)$.

Given a vertex $e \in E_\Delta$ of the Munn graph $\mathrm{M}(\Delta; \Sigma)$ at some $U$-invariant $\Delta \sse \Omega$, we denote by $\mathrm{M}(\Delta, e; \Sigma) \sse \mathrm{M}(\Delta; \Sigma)$ the connected component of $e$ and by $E_{\Delta,e}$ and $\Sigma_{\Delta,e}$ the set of its vertices and edges, respectively.
It will become apparent from the arguments below, but not stated explicitly, that the elements $e_\Delta u = (e_\Delta u \bar u) u (e_\Delta \bar u u)$ with $u \in \Sigma_{\Delta,e}$ generate the $\gD$-class of $e \in E(U)$ as a groupoid when $\Delta = (\dom(e))^U$ is chosen as in \cref{lem:gb-munn-vertices}.

\begin{definition}
  Let $\Delta \sse \Omega$ be a $U$-invariant set and $e \in E_\Delta$.
  We call a map $\gamma \colon E_{\Delta, e} \to U$ a \emph{basis} at $(\Delta, e)$ provided it satisfies the following conditions, wherein we write $\bar\gamma(f) = \overline{\gamma(f)}$.
  \begin{itemize}
    \item The element $\gamma(e)$ is idempotent (i.e., $\gamma(e) = \bar\gamma(e)$) and $\gamma(e) \geq \gamma(f)\bar\gamma(f)$ for all $f \in E_{\Delta, e}$.
    \item The element $\gamma(f)$ satisfies $\bar\gamma(f) e \gamma(f) = f$ and $\gamma(f) f \bar\gamma(f) = e$ for all $f \in E_{\Delta, e}$.
  \end{itemize}
\end{definition}

To construct a basis $\gamma$ at $(\Delta, e)$ we may proceed as follows.
First, let $\tilde e \in E(U)$ be the product $\prod u \bar u$ extending over all $u \in \Sigma_{\Delta,e}$ with $u \bar u \ge e$.
The idempotent $\tilde e$ will serve as $\gamma(e)$.
Next, for each other vertex $f \in E_{\Delta, e}$, we choose a path $(u_1, \dotsc, u_n)$ from $e$ to $f$ in $\mathrm{M}(\Delta, e; \Sigma)$ and set $\gamma(f) \coloneqq \tilde e u_1 \dotsc u_n$.
Using \cref{lem:gb-munn-paths}, it is easy to verify that $\gamma$ is as claimed.

Given a basis $\gamma$ at $(\Delta, e)$, we define $\lambda\colon \Sigma_{\Delta,e} \to U$ by $\lambda(u) = \gamma(e_\Delta u \bar u) \, u \, \bar\gamma(e_\Delta \bar u u)$.
Note that $\lambda(\bar u)$ is the inverse of $\lambda(u)$.
Moreover, we have $\lambda(u)\lambda(\bar u) = \lambda(\bar u )\lambda(u)$ by \cref{lem:gb-characterization} since, clearly, $\lambda(u)\lambda(\bar u) \gJ \lambda(\bar u )\lambda(u)$ and $\gamma(e) \geq \lambda(u)\lambda(\bar u), \lambda(\bar u )\lambda(u)$.

\begin{lemma}\label{lem:gb-munn-hclass}
  Let $e \in E(U)$, and let $\gamma$ be a basis at $(\Delta, e)$ where $\Delta = (\dom(e))^U \sse \Omega$.
  Then the $\gH$-class $U_e = \{s \in U \mid s\bar s = \bar s s = e\} \leq U$ is generated by $\Sigma_e \coloneqq \{ e \lambda(u) \mid u \in \Sigma_{\Delta, e}\} \sse U$.
\end{lemma}

\begin{proof}
  It is easy to verify that each $s = e\lambda(u)$ satisfies $s \bar s = \bar s s = e$.
  Conversely, let $s \in U$ with $s\bar s = \bar s s = e$ and write $s = u_1 \dotsc u_n$ with $u_1, \dotsc, u_n \in \Sigma$.
  Then $(u_1, \dotsc, u_n)$ is a path from $e$ to $e$ in $\mathrm{M}(\Delta, e; \Sigma)$ by \cref{lem:gb-munn-paths} as $\bar s e s = e$.
  In particular, $u_1, \dotsc, u_n \in \Sigma_{\Delta, e}$.

  Let $e = e_0, e_1, \dotsc, e_n = e \in E_{\Delta,e}$ be the vertices along the path $(u_1, \dotsc, u_n)$ and note that $e_0, e_1 \dotsc, e_n \in E(U)$ by \cref{lem:gb-munn-vertices}.
  We now compute 
  \[
    s = u_1 \dotsc u_n \ge e\gamma(e_0) \, u_1 \, \bar\gamma(e_1) \, e \gamma(e_1) \, u_2 \dotsc u_n \, \bar\gamma(e_n) = e\lambda(u_1)\, e\lambda(u_2) \dotsc e\lambda(u_n),
  \]
  wherein the left side equals $e_\Delta s$ and the right side equals $e_\Delta \lambda(u_1)\lambda(u_2) \dotsc \lambda(u_n)$.
  Therefore, we can then conclude that both sides of the inequality are equal by \cref{lem:gb-large}.
\end{proof}

\begin{lemma}\label{lem:gb-munn-idempotent}
  Let $e \in E(\ISym(\Omega))$.
   If $\hat e \in E(U)$ is minimal with $e \leq \hat e$, then there is a unique $e' \in E_\Delta$ with $e \leq e'$ where $\Delta = (\dom(e))^U$.
  Moreover, $\hat e = \prod \lambda(u)\lambda(\bar u)$ where the product extends over all $u \in \Sigma_{\Delta, e'}$ and $\lambda$ is obtained from some basis $\gamma$ at $(\Delta, e')$.
\end{lemma}

Note that $e \leq e' \leq \hat e$; hence, if $e= \hat e$, then also $e = e'$.

\begin{proof}
  Suppose that $\hat e \in E(U)$.
  Then $\dom(e) \sse (\dom(e))^U = \Delta$ and thus $e \leq e_\Delta$.
  Let $u \in \Sigma$ with $\hat e \leq u \bar u$.
  Then $e = e_\Delta e \leq e_\Delta \hat e \leq e_\Delta u \bar u$ and $u \bar u$ is $\Delta$-large; hence, $e' = e_\Delta u \bar u \in E_\Delta$.
  Conversely, if $u \in \Sigma_\Delta$ with $e \leq e_\Delta u \bar u$, then $e \leq u \bar u$ and, by minimality, $\hat e \leq u \bar u$.
  We obtain the inequality $e_\Delta \hat e \leq e_\Delta u \bar u$, which, by \cref{lem:gb-large}, is an equality.
  As such, $e' = e_\Delta \hat e$ is the unique vertex $e' \in E_\Delta$ with the property $e \leq e'$.

  To see that $\hat e$ can be written as the product $\hat e_\lambda$ of all $\lambda(u)\lambda(\bar u)$ with $u \in \Sigma_{\Delta, e'}$, we note that $\hat e = u_1 \dotsc u_n \bar u_n \dotsc \bar u_1$ for some $u_1, \dotsc u_n \in \Sigma$.
  Then $(u_1, \dotsc, u_n, \bar u_n, \dotsc, \bar u_1)$ is a path from $e'$ to $e'$ in $\mathrm{M}(\Delta; \Sigma)$ by \cref{lem:gb-munn-paths}; so $u_1, \dotsc, u_n \in \Sigma_{\Delta, e'}$.
  Let $e'_0, e'_1, \dotsc, e'_n, \dotsc, e'_1, e'_0 \in E_{\Delta, e'}$ be the vertices of $\mathrm{M}(\Delta; \Sigma)$ along this path.
  Inserting idempotents as in the proof of \cref{lem:gb-munn-hclass},
  \begin{align*}
    \hat e = u_1 \dotsc u_n \bar u_n \dotsc \bar u_1 &\geq \gamma(e'_0) u_1 \bar\gamma(e'_1) \dotsc u_n \bar\gamma(e'_n) \gamma(e'_n) \bar u_n \dotsc \gamma(e'_1)\bar u_1\bar \gamma(e'_0) \\
      &= \lambda(u_1) \dotsc \lambda(u_n)\lambda(\bar u_n) \dotsc \lambda(\bar u_1) = \lambda(u_1)\lambda(\bar u_1) \dotsc \lambda(u_n) \lambda(\bar u_n) \geq \hat e_\lambda
  \end{align*}
  where the final equality follows from \cref{lem:gb-idempotent_nesting} as $\lambda(u)\lambda(\bar u) = \lambda(\bar u)\lambda(u)$ for all $u \in \Sigma_{\Delta,e'}$.
  Since $e \leq e'\leq \hat  e_\lambda$, we have $\hat e \leq \hat e_\lambda$ by minimality of $\hat e$. 
  Hence, $\hat e = \hat e_\lambda$.
\end{proof}

\subsection{The Membership and Conjugacy Problems in $\vSI$}\label{sub:sis}

We now use the theoretical machinery developed in \cref{sub:munn} to show how the membership problem (and also the conjugacy problem to a certain extend) for finite strict inverse semigroups can be solved efficiently.
More precisely, we will show the following.

\begin{theorem}\label{thm:gb-complexity}
Let $\vH \sse \vG$ be a variety of finite groups. The problems $\dMemb[PB]{\vH\vee \vBS}$ and $\dConj[PB]{\vH \vee \vBS}$ are \LOGSPACE-reducible to $\dMemb[PB]{\vH}$ and $\dConj[PB]{\vH}$ for $\vH$, respectively.
\end{theorem}

Using the facts that $\dMemb[PB]{\vG}$ is in \NC (by {\cite{BabaiLS87}}, see \cref{pro:pbm-groups-membership}) and $\dConj[PB]{\vG}$ is in \NP (see \cref{pro:bpm-groups-conjugacy}), this implies the following upper bounds for $\vSI = \vG \vee \vBS$.

\begin{corollary}\label{cor:gb-complexity}
The problem	$\dMemb[PB]{\vSI}$ is in \NC and $\dConj[PB]{\vSI}$ is in \NP.
\end{corollary}

The other extreme, \ie taking $\vH$ to be the trivial variety $\vT$ in \cref{thm:gb-complexity}, yields the following upper bounds.
These are interesting because $\dEMemb[CT]{\vBS}$ and $\dEConj[CT]{\vBS}$ are already complete for \LOGSPACE{} as we will later show; see \cref{pro:membership_sl_hardness} and \cref{pro:conjugacy_sl_hardness}.

\begin{corollary}\label{cor:mem-BS}
The problems $\dMemb[PB]{\vBS}$ and $\dConj[PB]{\vBS}$ are in \LOGSPACE.
\end{corollary}

We prove \cref{thm:gb-complexity} using a sequence of short lemmas showing that the constructions from the previous subsection can actually be computed in \LOGSPACE.

\begin{lemma}\label{lem:gb-delta-sl}
	Given $\Sigma \sse \ISym(\Omega)$ and $X \sse\Omega$ the $\gen{\Sigma}$-invariant set $X^{\gen{\Sigma}}$ can be computed in \LOGSPACE.
\end{lemma}
\begin{proof}
	Define the (Schreier) graph $\Gamma$ with vertex set $\Omega$ and an edge from $x$ to $y \in \Omega$ whenever $y = x^u$ for some $u \in \Sigma$.
  Observe that this graph is undirected because $y = x^u$ if and only if $x = y^{\bar u}$.
  Now, $X^{\gen{\Sigma}}$ consists simply of the vertices with an incident edge and which are reachable from $X$.
  The latter can be checked using an oracle for \dUGAP. 
\end{proof}

\begin{lemma}\label{lem:gb-mann-sl}
	Given $\Sigma \sse \ISym(\Omega)$ and $\Delta$, the Munn graph $\mathrm{M}(\Delta; \Sigma)$ can be computed in \LOGSPACE.
\end{lemma}

\begin{proof}
  To determine the edge set $\Sigma_\Delta$ of $\mathrm{M}(\Delta; \Sigma)$, we check whether $(\dom(u)\cap \Delta)^U = \Delta$ for each $u \in \Sigma$ using \cref{lem:gb-delta-sl}.
   Now, \cref{def:gb-munn} immediately gives us $E_\Delta$ and $\mathrm{M}(\Delta; \Sigma)$.
\end{proof}

\begin{lemma}\label{lem:gb-basis-sl}
	Given a Munn graph $\mathrm{M}(\Delta; \Sigma)$ and $e \in E_\Delta$, a basis $\gamma \colon E_{\Delta, e} \to U$ at $(\Delta, e)$ represented as a list $(f,\gamma(f))_{f \in E_{\Delta, e}}$, where $U = \gen\Sigma$, can be computed in \LOGSPACE.
\end{lemma}

\begin{proof}
	The alphabet $\Sigma_{\Delta, e}$ needed for the definition of $\gamma$ can be found by computing the connected component of $\mathrm{M}(\Delta; \Sigma)$ containing $e$, which can be done using an \dUGAP oracle. 
	Next, to compute $\gamma$, we use \cite[Lemma 2.4]{NisanT95} to find a path $(u_1, \dots, u_n)$ from $e$ to some arbitrary $f\in E_{\Delta,e}$.
	Note that computing the product of a sequence of elements $u_1, \dots, u_n \in \ISym(\Omega)$ can be done in \LOGSPACE by evaluating $x^{u_1 \cdots u_n}$ for each $x \in \Omega$ separately.
\end{proof}

\begin{lemma}\label{lem:gb-hat-e-sl}
  Given $\Sigma \sse \ISym(\Omega)$ and $e \in E(\ISym(\Omega))$ the minimal $\hat e \in E(U) \cup \{1\}$ with $e \leq \hat e$ can be computed in \LOGSPACE{} and, as part of the computation, one can decide whether $\hat e \in U = \gen\Sigma$.
\end{lemma}

\begin{proof}
	First, compute $\Delta = \dom(e)^U$ using \cref{lem:gb-delta-sl} and the Munn graph $\mathrm{M}(\Delta; \Sigma)$ using \cref{lem:gb-mann-sl}.
	Next find some  $e' \in E_\Delta$ with $e \leq e'$. If no such $e'$ exists, then $\hat e = 1$ and $\hat e \not\in U$.
	Otherwise, compute a basis  $\gamma \colon E_{\Delta, e'} \to U$  at $(\Delta, e')$ by \cref{lem:gb-basis-sl}.	
	Finally, to compute $\hat e$, we apply the formula from \cref{lem:gb-munn-idempotent}, which uses the already computed basis $\gamma$.	
\end{proof}

\begin{proof}[Proof of \cref{thm:gb-complexity}]
  We first consider the membership problem.
  On the input of $\Sigma \sse \ISym(\Omega)$ with $U = \langle \Sigma \rangle \in \vH \vee \vBS$ and $t \in \ISym(\Omega)$, let $e = t \ov{t}$ and $f = \ov{t} t$.
  We then proceeds as follows.

Compute $\Delta = \dom(e)^U = \dom(f)^U$ using \cref{lem:gb-delta-sl}.
 If $ \dom(e)^U \neq \dom(f)^U$, then $t \not\in U$ and we can output a fixed negative instance.
 Next, we compute $\hat e$ and $\hat f$ (\cref{lem:gb-hat-e-sl}) and verify that $ e = \hat e \in E_\Delta$ and $f = \hat f \in E_\Delta$ (if not, then $t \not\in  U$).
  Let $\gamma \colon E_{\Delta, e} \to U$ be a basis at $(\Delta, e)$, which we can compute by \cref{lem:gb-basis-sl}.
  Finally, using the basis $\gamma$ we compute a generating set $\Sigma_{e}\sse \ISym(\Omega)$ of the $\gH$-class $U_{ e} = \{s \in U \mid s\bar s = \bar s s =  e\}$ as in \cref{lem:gb-munn-hclass} and reduce to the question of whether $t' \coloneqq t \bar \gamma(f) \in \langle \Sigma_{e} \rangle$~-- an instance of \dMemb[PB]{\vH}.

  If $t' \in U_e$, then $t = t' \gamma(f) = t \bar \gamma(f) \gamma(f) \in U$.
	For the other direction, assume that $t \in U$.
	Then clearly $e, f \in U$ so $e = \hat e$ and $f = \hat f$.
	Moreover, $e$ and $f$ are vertices of a connected component of $\mathrm{M}(\Delta; \Sigma)$ by \Cref{lem:gb-munn-vertices}.
  Finally, $t' \in U_{e}$ as $t \bar \gamma(f) \gamma(f) \bar t = \gamma(f) \bar t t \bar \gamma(f)  = e$.

	\medskip

  We now turn our attention to the conjugacy problem.
	The input comprises $U = \langle \Sigma \rangle$ as above and two elements $s, t \in \ISym(\Omega)$.
	The question is whether there exists some $u \in U^1$ such that $\ov u s u = t$ and $u t \ov u = s$.
  We assume throughout that $s \neq t$ (otherwise, we reduce to a fixed positive instance).
	Let $e = s\ov s \vee \ov s s \in  E(\ISym(\Omega))$ and $f = t\ov t \vee \ov t t\in E(\ISym(\Omega))$.
	
  Compute $\hat e$ and $\hat f$ (as above), as well as $\hat \Delta = \dom(\hat e)^U = \dom(\hat f)^U$. 
  Note that we do not require that $e = \hat e$ or $f = \hat f$.
  If $\dom(\hat e)^U \neq \dom(\hat f)^U$ or $\hat e, \hat f \not\in U$ or, similarly, if $\hat e$ and $\hat f$ are not connected in $\mathrm{M}(\hat\Delta; \Sigma)$, then $s \not\sim_U t$.
  Next, compute a basis $\gamma \colon E_{\hat\Delta, \hat e} \to U$ at $(\hat\Delta,\hat e)$ and a corresponding generating set $\Sigma_{\hat e}$ of the $\gH$-class $U_{\hat e}$ of $\hat e$ as in \cref{lem:gb-munn-hclass}. 
	Finally, reduce to the question of whether $s$ is conjugate to $t' \coloneqq \gamma(\hat f) t \bar \gamma(\hat f)$ in $\langle \Sigma_{\hat e} \rangle$~-- an instance of \dConj[PB]{\vH}.

  If $s \sim_{U_{\hat e}} t'$, then $s \sim_U t$ as $t \sim_U t'$.
	For the other direction, assume that $s$ and $t$ are conjugate by some element $u \in U$, \ie $\bar us u = t$ and $u t \bar u = s$.
  Then $\hat e$ and $\hat f$ are also conjugate by the element $u$.
  To see this, note that $\bar u \hat e u \ge f$.
  Hence, $\bar u \hat e u \ge \hat f$ (by minimality of $\hat f$) and $u \hat f \bar u \ge \hat e$ (by symmetry).
  The chain of inequalities $\hat f = \bar u u \hat f \bar u u \ge \bar u \hat e  u \ge \hat f$ then finally shows that $\bar u \hat e  u = \hat f$ and $u \hat f \bar u = \hat e$ (by symmetry).
	It now follows from \Cref{lem:gb-munn-vertices} that $\hat e$ and $\hat f$ are vertices of a connected component of $\mathrm{M}(\hat \Delta; \Sigma)$.
  Writing $v = \bar\gamma(\hat f)$, we can easily verify that $uv$ conjugates $s$ to $t'$, \ie $\ov{v}\ov{u} s uv = \ov v t v = t'$ and $uv t' \ov{v}\ov{u} = uv (\ov v t v) \ov{v}\ov{u} = s$.
  As clearly $\hat e s \hat e = s$, the element $\hat e uv$ also conjugates $s$ to $t'$.
  Moreover, $\hat e u v \in U_{\hat e}$ and, as such, $s \sim_{U_{\hat e}} t'$.
\end{proof}

\section{Inverse Semigroups in the Cayley Table Model}\label{sec:SL-Cayley}

In this section we complete the proof of \cref{thm:main-CT}, our main dichotomy theorem for Cayley table model, which is restated here for the readers convenience.

\begin{theorem}\label{thm:rst:main-CT}
	Let $\vV$ be a variety of finite inverse semigroups. 
	\begin{itemize}
		\item If $\vV \sse \vCl$, then $\dMemb[CT]{\vV}$ and $\dConj[CT]{\vV}$ are in \NPOLYLOGTIME and in \LOGSPACE.
		\item If $\vV \not\sse \vCl$, then $\dMemb[CT]{\vV}$ and $\dConj[CT]{\vV}$ are \LOGSPACE-complete under \ACz-reductions.
	\end{itemize}
\end{theorem}

Recall that the variety $\vCl$ of finite Clifford semigroup is defined by the identity $\ov x x = x \ov x$; it is the smallest variety containing all groups and semilattices. 
In \cref{sec:Clifford-CT} we have seen that \dMemb[CT]{\vCl} and \dConj[CT]{\vCl} are in \NPOLYLOGTIME.
Completing the proof of \cref{thm:rst:main-CT} consists of two steps: in \cref{sec:CT-SL} we establish \LOGSPACE-algorithms for these problems, and in \cref{sec:SL-hard} we show that the problems are hard for \LOGSPACE given that $\vV \not\sse \vCl$.
Before we go into the details of our proof, let us explore an interesting consequence of \cref{thm:rst:main-CT}.

\begin{corollary}\label{cor:short-SLP-char}
  Let $\vV$ be a variety of finite inverse semigroups.
  The following are equivalent.
  \begin{itemize}
    \item The class $\vV$ comprises only Clifford semigroups, i.e., $\vV \sse \vCl$.
    \item The class $\vV$ admits polylogarithmic SLPs.
    \item The problem $\dMemb[CT]{\vV}$ is in \NPOLYLOGTIME.
    \item The problem $\dMemb[CT]{\vV}$ is in \qACz.
  \end{itemize}
\end{corollary}
\begin{proof}
	By \cref{lem:clifford-slp}, any variety $\vV \sse \vCl$ admits polylogarithmic SLPs.
	By \cref{lem:SLP-npolylogtime}, if $\vV$ admits polylogarithmic SLPs, then $\dMemb[CT]{\vV}$ is in $\NPOLYLOGTIME \sse \qACz$.

	Finally, if $\vV \not\sse \vCl$, then the problem \dMemb[CT]{\vV} is \LOGSPACE-complete under \ACz-reductions by \cref{thm:rst:main-CT}. 
	Hence, this problem cannot be solved in \qACz as for example \prob{parity} can be solved in \LOGSPACE but not in \qACz \cite{FurstSS84,Hastad86}.
\end{proof}

Note that the equivalence of the first two points of \cref{cor:short-SLP-char} can also be proved directly. 
Indeed, 
\ifAnonimous Fleischer\else the first author\fi\ 
determined in his dissertation \cite{Fleischer19diss} the maximal varieties of finite (not necessarily inverse) \emph{monoids} admitting polylogarithmic SLPs, viz.\ the varieties of finite Clifford monoids and of finite commutative monoids.
However, the situation for semigroups is considerably more intricate.
Therefore, we point out the following open problem.

\begin{question}
  Which varieties of arbitrary finite semigroups admit polylogarithmic SLPs?
\end{question}

\subsection{Membership and Conjugacy in \LOGSPACE}\label{sec:CT-SL}

It is known that the membership problem for semigroups in the Cayley table model belongs to \NL \cite{JonesLL76}.
This immediately carries over to inverse semigroups.
In this section we go further and show that this can improved to \LOGSPACE for inverse semigroups.
Recall that \NL is intricately related to \emph{directed} graph accessibility, while \SL{} = \LOGSPACE by Reingold's result \cite{Reingold08} corresponds to \emph{undirected} graph accessibility.
This difference in complexity is explained by the observation that strong connected components of the (right) Cayley graph of an inverse semigroup are actually undirected graphs.
To formalize this, let us define the decision problem \dRequiv[CT]{}.

\begin{decproblem}
	\iitem An inverse semigroup $S$ as Cayley table, a subset $\Sigma \sse S$, and elements $s, t \in S$.
	\qitem Is $s \gR[U] t$ where $U = \gen{\Sigma}$?
\end{decproblem}

	Recall that $x \gR[U] y$ for $x,y \in S$ if and only if there are $r,s \in  U $ with $xr=y$ and $ys = x$.

\begin{proposition}\label{lem:conj_reachability}
	The problems \dRequiv[CT]{} and $\dConj[CT]{}$ are in \LOGSPACE.
\end{proposition}

\begin{proof}
	Both problems are essentially reachability in an undirected graph.
	Indeed, define the undirected graph $\Gamma$ with vertex set $S$ and an edge between $x$ and $y$ for $x,y \in S$ whenever there is some $u \in \Sigma$ with $xu = y$ and $x = y \ov u$. Clearly, $\Gamma$ can be computed in \LOGSPACE.

	By \cref{lem:green-semi_rel}, if $xu \gR[\gen{\Sigma}] x$ for $u \in \Sigma$ and $x \in S$, then $x$ and $xu$ are connected by an edge in $\Gamma$ (and the converse is obviously true). In particular, the connected components of $\Gamma$ are precisely the strongly connected components of the right Cayley graph of $S$ with respect to the generating set $\Sigma$.
	Therefore, reachability in $\Gamma$ is exactly the question whether  $s \gR[\gen{\Sigma}] t$.
	
	To decide whether or not $s \sim_{\gen{\Sigma}} t$, we proceed in the same way but in the graph with edges $\{x,y\}$ whenever  there is some $u \in \Sigma$ such that $\ov u xu = y$ and $ x = uy\ov u$.
\end{proof}

\begin{proposition}\label{thm:InvSGSL}
	The problem \dMemb[CT]{} is in \LOGSPACE.
\end{proposition}

\begin{proof}
	We are given an inverse semigroup $S$, a subset $\Sigma \sse S$ and an element $t \in S$ and we want to decide whether $t \in U = \gen{\Sigma}$.
	To this end, we describe a $\LOGSPACE$-algorithm using oracle calls to \dUGAP based on \cref{lem:conj_reachability}.

  Without loss of generality, we assume that both $S$ and $U$ contain a neutral element $1$ (by simply adjoining such an element) and that $\Sigma$ is closed under formation of inverses.
	Our algorithm then proceeds as follows.	
  \begin{samepage}
  	\begin{algorithmic}[1]
  		\State $x \gets 1$
  		
  		\While{$\exists y \in S$, $u \in \Sigma$ \textbf{with} $x \gR[U] y$ \textbf{and} $yu \ov u \neq y$ \textbf{and} $yu\ov u \ov y\, t = t$}
  		\State 	$x \gets yu$
  			\EndWhile
  		\If{$x \gR[U] t$}
  			\State \textbf{return true}
  		\EndIf
  		\State \textbf{return false}  	
  	\end{algorithmic}
  \end{samepage}

	The tests whether $x \gR[U] y$ can be done in \LOGSPACE by \cref{lem:conj_reachability} (using oracle calls to \dUGAP).
	Moreover, the tests whether there exist $y \in S$, $u \in \Sigma$ meeting the conditions in line 2, can be done by iterating over all such elements checking whether the conditions are satisfied.

	Throughout, we keep the invariant that $x\in U$.
	Therefore, if our algorithm outputs \textbf{true}, then, indeed, $t \in U$.
	On the other hand, let $t = u_1 \cdots u_n \in U$ with $u_i \in \Sigma$.
	The idea is that the algorithm finds this (or a slightly modified) sequence.
    Observe that besides $x \in U$ we maintain the invariant $x \ov x t = t$ (\ie $x \gRge t$).

  Next, observe that, by \cref{lem:green-semi_rel}, $yu \ov u \neq y$ means that $y$ is not $\gR[U]$-equivalent to $yu$ or, more specifically, that $y \gRgt[U] yu$.
	As such, the while loop can be executed only finitely (indeed, at most $\abs{U}$) many times.
	Therefore, we can proceed in the following by induction on the number of times that line 3, the body of the while loop, is being executed.
	
	If $x \in U$ with $x \ov x t = t$ and $t \in U$, then either $x \gR[U] t$ or there is some $j \in \{1, \dots, n\}$ with $y = x \ov x  u_1\cdots u_j  \gR[U] x$ but $x \ov x  u_1\cdots u_{j+1}  \gRlt[U] x$. 
  	In the former case, there are no $y$ and $u$ meeting the conditions in line 2; hence, the algorithm terminates and outputs \textbf{true}.
	In the latter case, we have $x \gR[U] y$ and $yu_{j+1} \ov u_{j+1} \neq y$ and $y u_{j+1} \ov u_{j+1}  \ov y t = t$.
	As such, the while loop will be executed at least one more time and, by induction, the algorithm will therefore answer correctly that $t \in U$.
\end{proof}

\subsection{Hardness of Membership and Conjugacy for \LOGSPACE{}}\label{sec:SL-hard}

We now turn to hardness results for the (idempotent) membership and conjugacy problem for inverse semigroups in the Cayley table model.
Recall that a variety of finite inverse semigroups $\vV$ satisfies $\vV \not\sse \vCl$ if and only if $\vBS \sse \vV$, \ie $\vV$ contains the combinatorial Brandt semigroup $B_2$; see the second item of \cref{pro:varieties}.

\begin{proposition}\label{pro:conjugacy_sl_hardness}
  Let $\vV$ be a variety of finite inverse semigroups and suppose that $\vBS \sse \vV$.
	Then the problem $\dEConjS[CT]{\vV}$ is \LOGSPACE-hard under \ACz-reductions.
\end{proposition}

\begin{proof}
  To show hardness for $\LOGSPACE{} = \SL{}$, we reduce from the problem \dUGAP.
	Given an undirected graph $G = (V, E)$ and vertices $s,t \in V$, we proceed as follows.
	
	Let $S = B(V)$ be the Brandt semigroup on $V$ (recall that $B_2 \in \vV$ implies $B(V) \in \vV$ by \cref{lem:b2_implies_bn}).
	Given $x,y \in V$, we denote by $u_{xy}$ the unique element of $S$ mapping $x$ to $y$.
	Note that each non-zero element of $S$ is of this form, and the non-zero idempotents are precisely the elements $e_x \coloneqq u_{xx}$.
	Crucially, for fixed idempotents $e_x$ and $e_y$, the equation $\bar u e_x u = e_y$ has exactly one solution in $S$, namely the element $u = u_{xy}$.
  	In particular, $e_x \sim_S e_y$.
	
	Let $\Sigma = \{ e_x \mid x \in V \} \cup \{ u_{xy} \mid xy \in E \}$.
	An element $u_{xy}$ is contained in $U = \langle \Sigma \rangle \leq S$ if and only if the vertices $x$ and $y$ are connected by a path in $G$.
	In particular, $e_s \sim_U e_t$ holds if and only if $s$ and $t$ are connected by a path in $G$.
	Since the Cayley table of $S$, the set $\Sigma \subseteq S$, and the idempotents $e_s, e_t \in E(S)$ can be computed in \ACz from $G = (V,E)$ and $s,t \in V$, we conclude that the (restricted) idempotent conjugacy problem for $\vV$ is \LOGSPACE-hard.
\end{proof}

This argument also shows that the membership problem for $\vV$ is \LOGSPACE-hard in the Cayley table model as we can simply ask whether or not $u_{st} \in U$.
Even more, the following holds.

\begin{proposition}\label{pro:membership_sl_hardness}
	Let $\vV$ be a variety of finite inverse semigroups and suppose that $\vBS \sse \vV$.
	Then the idempotent membership problem $\dEMembS[CT]{\vV}$ is \LOGSPACE-hard. 
\end{proposition}

Together with \cref{pro:conjugacy_sl_hardness}, \cref{pro:membership_sl_hardness} establishes the hardness part of \cref{thm:rst:main-CT}.
Our proof of \cref{pro:membership_sl_hardness} is a slight variation of the proof of \cref{pro:conjugacy_sl_hardness}.
It is based on the observation that certain instances of the idempotent conjugacy problem reduce to the idempotent membership problem for a slightly larger inverse semigroup.

\begin{lemma}\label{lem:idem_memb_from_idem_conj}
	Let $U$ be an inverse semigroup and $e_s,e_t \in E(U)$.
	Further, let $U' \leq U \times Y_2$ be generated by $(e_s, 0)$ and $U \times \{1\}$.
	Then $e_s \gJge e_t$ holds if and only if $(e_t, 0) \in U'$.
\end{lemma}
\begin{proof}
	If $e_s \gJge e_t$, then $e_t = u e_s v$ for some $u,v \in U^1$; hence, $(e_t, 0) = (u, 1) (e_s, 0) (v, 1) \in U'$.
	Conversely, if $(e_t, 0) \in U'$ holds, then $(e_t, 0) = u_1u_2 \dotsc u_n$ for some $u_i \in \{(e_s, 0)\} \cup U \times \{1\}$.
	Clearly, at least one of the factors $u_i$ must be equal to $(e_s, 0)$.
	By replacing $(e_s, 0)$ with $(e_s, 1)$ for all but one of the factors $u_i$, and using the fact that $U \times \{1\}$ is a subsemigroup of $U'$, we obtain $(e_t, 0) = (u, 1) (e_s, 0) (v, 1)$ for some $u,v \in U^1$; hence, $e_s \gJge e_t$ as $e_t = u e_s v$.
\end{proof}

\begin{proof}[Proof of {\cref{pro:membership_sl_hardness}}]
	Given an undirected graph $G = (V,E)$ and $s,t \in V$, we let $U \leq S$ and $\Sigma$ be as in the proof of \cref{pro:conjugacy_sl_hardness}.
	We consider $S' = S \times Y_2$ and $U' = \langle \Sigma' \rangle \leq S'$ where $\Sigma' = \{ (e_s, 0) \} \cup \Sigma\times\{1\}$.
	Note that $U',S' \in \vV$ since $S \in \vV$ and $Y_2 \in \vV$.
	By \cref{lem:idem_memb_from_idem_conj}, the idempotent $(e_t, 0) \in S'$ is contained in $U'$ if and only if $e_s \gJge[U] e_t$.
	By \cref{lem:idem_conj_is_green_j}, this is equivalent to the existence of some $u \in U^1$ with $e_t = \bar u e_s u$.
	Since we already know that the latter holds if and only if $s$ and $t$ are connected by a path in $G$, this completes the reduction from the undirected graph reachability problem to the idempotent membership problem.
\end{proof}

\section{Inverse Semigroups in the Partial Bijection Model}\label{sec:PSPACE-hard}

In this section we finally complete the proof of our dichotomy theorem regarding the membership and conjugacy problem for inverse semigroups in the partial bijection model.

\begin{theorem}[\cref{thm:main-PB}]\label{thm:pbm-dichotomy}
	Let $\vV$ be a variety of finite inverse semigroups. 
	\begin{itemize}
		\item If $\vV \sse \vSI$, then $\dMemb[PB]{\vV} $ is in \NC and $\dConj[PB]{\vV}$ is in \NP.
		\item If $\vV \not\sse \vSI$, then $\dMemb[PB]{\vV}$ and $\dConj[PB]{\vV}$ are \PSPACE-complete.
	\end{itemize}
\end{theorem}

As outlined in \cref{sec:results}, the case $\vV \sse \vSI$ can be further refined as follows implying actually an \ACz-vs.-\NC-vs.-\PSPACE-complete trichotomy for $\dMemb{}$.
\begin{itemize}
	\item If $\vV \sse \vSl$, then $\dMemb[PB]{\vV} $ and $\dConj[PB]{\vV}$ are in \ACz.
	\item If $\vV = \vBS$, then $\dMemb[PB]{\vV} $ and $\dConj[PB]{\vV}$ are \LOGSPACE-complete.
	\item If $\vV \not\sse \vBS$, then $\dMemb[PB]{\vV} $ is in \NC and $\dConj[PB]{\vV}$ is in \NP; both are hard for \LOGSPACE.
\end{itemize}

Indeed, $\dMemb[PB]{\vSl} $ is considered in  \cite{BeaudryMT92} and shown to be in \ACz.
Note that the conjugacy problem for semilattices is trivially in \ACz as in a semilattice $S$ two elements are conjugate if and only if they are equal (recall that $S = E(S)$ and that $S$ is $\gJ$-trivial).

If $\vV \not\sse \vSl$, then either $\vV$ contains the Brandt semigroup $B_2$ or some non-trivial group. 
In the first case, we get \LOGSPACE-hardness by \cref{thm:main-CT} as by the Preston-Wagner Theorem \cite{Preston54,Wagner52} every inverse semigroup given as Cayley table can be interpreted as an inverse semigroup in the partial bijection model.
In the second case, \LOGSPACE-hardness of $\dMemb[PB]{\vV} $ and $\dConj[PB]{\vV}$ is established in \cref{lem:nontrivial-group-hard}.
Thus, both problems are \LOGSPACE-hard whenever $\vV \not\sse \vSl$.
The fact that $\dMemb[PB]{\vBS}$ and $\dConj[PB]{\vBS}$ are contained in \LOGSPACE is the content of \cref{cor:mem-BS}.

For proving \cref{thm:pbm-dichotomy}, recall that either $\vV \sse \vSI$ or $\vBM \sse \vV$ for each variety $\vV$ of finite inverse semigroups by \cref{pro:varieties}.
With the first part of \cref{thm:pbm-dichotomy} covered by \cref{cor:gb-complexity}, it therefore suffices to show that the problems \dMemb[PB]{\vBM} and \dConj[PB]{\vBM} are \PSPACE-complete.
Moreover, since these problems are clearly contained in \PSPACE{}, it suffices to prove their hardness for it.
We will do so by a suitable reduction in \cref{sub:pbm-hardness}.

Slight variations of this reduction will then allow us to derive hardness results for the intersection non-emptiness problem as well as the subpower membership problem in \cref{sub:iai-hardness} and \cref{sub:subpower}, respectively.

\subsection{Hardness of Membership and Conjugacy for \PSPACE}\label{sub:pbm-hardness}

As mentioned in the introduction to this section, our goal is to prove \PSPACE-hardness of the membership and conjugacy problems for $\vBM$ in the partial bijection model.
In fact, we will show that this even holds for the idempotent variants of these problems.

\begin{theorem}\label{thm:pbm-hardness}
  The problems \dEMembS[PB]{\vBM} and \dEConjS[PB]{\vBM} are \PSPACE-complete.
\end{theorem}

These problems are clearly contained in \PSPACE{}.
We prove their \PSPACE-hardness by a reduction from the decision problem \dNCL of \emph{non-determinstic constraint logic} (\emph{NCL}), which was introduced by Hearn and Demaine~\cite{HearnD05} and which we briefly describe in the following.

An \emph{NCL machine} $\Gamma$ is an edge-weighted simple undirected graph, every vertex of which has degree three, and every edge of which has weight one or two.
A \emph{configuration} of the NCL machine is an orientation of all edges such that the sum of the incoming edge weights (in-flow) at every vertex is at least two.
We will denote the set of all configurations by $\cC(\Gamma)$.
Two configurations are related by a \emph{transition} if they differ in the orientation of a single edge.

The decision problem \dNCL, in the \emph{configuration-to-configuration} variant, is given as follows.

\begin{decproblem}
  \iitem An NCL machine $\Gamma$, and two configurations $C_s, C_t \in \cC(\Gamma)$.
  \qitem Are $C_s$ and $C_t$ related by a sequence of transitions?
\end{decproblem}

This problem is essentially a compressed version of the accessibility problem for undirected graphs, where accessibility is decided on the implicitly given graph of configurations and transitions between these.
The problem \dNCL{} is complete for \PSPACE{}~\cite[Theorem 5]{HearnD05}.

Crucial to our cause is the observation that the validity of any configuration can be verified locally, as the minimum in-flow constraint has to be satisfied at each vertex individually. 
Moreover, transitions are also local in the sense that each transition affects only a single edge and, therefore, only the two vertices incident with that edge.

We call an orientation of all edges incident with a fixed vertex $v$ a \emph{local configuration} at $v$ provided that the minimum in-flow constraint at $v$ is satisfied.
The set of all local configurations at $v$ will be denoted by $\cC(\Gamma, v)$.
Note that $\lvert \cC(\Gamma, v) \rvert \leq 7$ as $v$ has degree three and at least one edge needs to be oriented towards $v$ to satisfy the in-flow constraint.

Given an orientation $O$ of all edges of $\Gamma$, we denote its restriction to an orientation of the edges incident with $v$ by $O\vert_v$.
Note that an orientation $O$ is a configuration, i.e., $O \in \cC(\Gamma)$, if and only if each restriction $O\vert_v$ is a local configuration, i.e., $O\vert_v \in \cC(\Gamma, v)$ for all $v \in V(\Gamma)$.

\begin{proof}[Proof of \cref{thm:pbm-hardness}]
  We first reduce the problem \dNCL to the problem \dEConjS[PB]{\vBM}.
  This shows that the latter problem is hard for \PSPACE{} and, therefore, \PSPACE-complete.

  Given an NCL machine $\Gamma$, we associate to it the inverse semigroup $S_\Gamma = \prod_{v} B^1(\cC(\Gamma, v))$ where the direct product extends over all vertices $v \in V(\Gamma)$.
  We identify $S_\Gamma$ with an inverse subsemigroup of $\ISym(\Omega_\Gamma)$ where $\Omega_\Gamma = \bigsqcup_v \cC(\Gamma, v)$.
  Note that $\lvert \Omega_\Gamma \rvert \in \bigO(\lvert V(\Gamma) \rvert)$.

  Each configuration $C \in \cC(\Gamma)$ has an idempotent $e(C) \in E(S_\Gamma)$ canonically associated to;
  the projection of $e(C)$ onto a factor $B^1(\cC(\Gamma, v))$ is given by the idempotent at $C\vert_v \in \cC(\Gamma, v)$.
  Note that any two such idempotents are conjugate in $S_\Gamma$.
  We now describe how to encode a transition by a corresponding element of $S_\Gamma$.
  Suppose given an oriented edge $o$ from $v_1$ to $v_2$, say, and local configurations $c_i \in \cC(\Gamma, v_i)$ with $o \in c_i$ for $i = 1,2$.
  Let $o'$ denote the reversal of $o$, and let $c'_i$ denote the result of replacing $o$ by $o'$ in $c_i$.
  Provided that $c'_1 \in \cC(\Gamma, v_1)$ and $c'_2 \in \cC(\Gamma, v_2)$, we define $u(c_1, c'_2) \in S_\Gamma$ as follows.
  The projection of $u(c_1, c'_2)$ onto $B^1(\cC(\Gamma, v))$ where $v = v_i$ is the unique partial bijection in $B^1(\cC(\Gamma, v_i))$ with $c_i \mapsto c'_i$, and for $v \not\in \{ v_1, v_2 \}$ the projection of $u(c_1, c'_2)$ onto the factor $B^1(\cC(\Gamma, v))$ is the identity element.
  Note that we have $\overline{u(c_1, c'_2)} = u(c'_2, c_1)$ and that $u(c_1, c'_2)$ is defined if and only if $u(c'_2, c_1)$ is.

  Let $U_\Gamma \leq S_\Gamma$ be the inverse subsemigroup generated by the set $\Sigma_\Gamma \sse S_\Gamma \leq \ISym(\Omega_\Gamma)$ of all elements $u(c_1, c'_2)$ as above.
  Note that $\lvert \Sigma_\Gamma \rvert \in \bigO(\lvert V(\Gamma)\rvert)$.
  We claim that $e(C_s) \sim_{U_\Gamma} e(C_t)$ for two configurations $C_s, C_t \in \cC(\Gamma)$ if and only if these configurations can be transformed into one another by a sequence of transitions.
  To see this, consider an element of the form $\bar ueu$ where $e = e(C)$ for some $C \in \cC(\Gamma)$ and with $u = u(c_1, c'_2) \in U_\Gamma$ as above.
  Then either $C\vert_{v_i} = c_i$ for $i=1,2$, in which case $\bar u e u = e(C')$ with $C' \in \cC(\Gamma)$ obtained from $C$ by reversal of $o$, or else $\bar u e u \in I_\Gamma$ where $I_\Gamma \subseteq S_\Gamma$ is the ideal comprising all elements with at least one projection equal to $0 \in B^1(\cC(\Gamma, v))$.
  Since $e(C) \not\in I_\Gamma$ for all $C \in \cC(\Gamma)$, this proves the claim.

  Given an instance of the problem \dNCL{} comprising an NCL machine $\Gamma$ and $C_s, C_t \in \cC(\Gamma)$, we associate with it the instance of the problem \dEConjS[PB]{\vBM} comprising $\Sigma \coloneqq \Sigma_\Gamma \sse \ISym(\Omega_\Gamma)$ and idempotents $s \coloneqq e(C_s), t \coloneqq e(C_t) \in E(S_\Gamma) \sse E(\ISym(\Omega_\Gamma))$ as above.
  Clearly, the latter can be computed in polynomial time from the former.
  As the problem \dNCL{} is hard for \PSPACE{} under polynomial-time many-one reductions, so is the problem \dEConjS[PB]{\vBM}.

  \medskip

 Finally, to see that the problem \dEMembS[PB]{\vBM} is also hard for \PSPACE, let $S'_\Gamma = S_\Gamma \times Y_2$ and $U'_\Gamma$ be its inverse subsemigroup generated by $\Sigma_\Gamma \times \{1\}$ and the elements $(e(C_s), 0), (e(C_t), 1) \in S'_\Gamma$.
  Then $(e(C_t), 0) \in U'_\Gamma$ if and only if $e(C_s) \gJge[U_\Gamma] e(C_t)$ by \cref{lem:idem_memb_from_idem_conj}.
  Since $e(C_s) \gJ e(C_t)$ holds in $S_\Gamma$, we then have $(e(C_t), 0) \in U'_\Gamma$ if and only if $e(C_s) \gJ[U_\Gamma] e(C_t)$ if and only if $e(C_s) \sim_{U_\Gamma} e(C_t)$ by \cref{lem:green-semi_rel} and by \cref{lem:idem_conj_is_green_j}, respectively.

  We finally note that $U'_\Gamma \leq S'_\Gamma$ can be realized as an inverse subsemigroup of $\ISym(\Omega_\Gamma \sqcup \{\ast\})$ and that $U'_\Gamma \leq S'_\Gamma \in \vBM$.
  Hence, the problem \dNCL{} also admits a polynomial-time many-one reduction to the problem \dEMembS[PB]{\vBM}.
  As such, the latter is hard for \PSPACE.
\end{proof}

\subsection{The Intersection Non-Emptiness Problem}\label{sub:iai-hardness}

Next, let us explore a variant of the reduction from the proof of \cref{thm:pbm-hardness} to show further hardness results for the intersection non-emptiness problem for inverse automata.

For a definition of deterministic finite automata (DFA) we refer the reader to  standard textbooks, \eg \cite{HU}.
We denote the accepted language of a DFA $\cA$ by $\cL(\cA)$.

An \emph{inverse} automaton is a DFA $\cA = (Q,\Sigma, \delta, q_0, F)$ with a \emph{partially} defined transition function $\delta\colon Q  \times \Sigma \to Q$ such that the following conditions are satisfied.
\begin{itemize}
  \item For each $a \in \Sigma$, the partial map $\delta_a\colon Q \to Q$ with $q \mapsto \delta(q,a)$ is injective on its domain. 
  \item For each $a \in \Sigma$, there is a word $w_a \in \Sigma^\ast$ such that $\delta_a \delta_{w_a}$ is the identity on $\dom(\delta_a)$, where $\delta_{w_a} = \delta_{a_1}\delta_{a_2} \dotsc \delta_{a_n}$ is the partial map induced by $w_a = a_1 a_2 \dotsc a_n$.\footnote{We view the maps $\delta_a$ as elements of $\ISym(Q)$ and, as such, compose them as operators acting on the right.}
\end{itemize}
\begin{remark}
  We will only encounter inverse automata over some alphabet endowed with an involution $a \mapsto \ov a$ where one can take $w_a = \ov a$ in the second condition above.
\end{remark}

The \emph{intersection non-emptiness problem} for a class $\cX$ of automata (e.g., $\cX = \textsc{dfa}$ or $\cX = \textsc{ia}$), which we denote by \dIEmpty{$\cX$}, is the decision problem defined as follows.

\begin{decproblem}
  \iitem Automata $\cA_1, \dotsc, \cA_k \in \cX$.
  \qitem Is $L = \bigcap_{i=1}^k\cL(\cA_i)$ non-empty?
\end{decproblem}

The problems \dIEmpty{dfa} and \dIEmpty{ia} are \PSPACE-complete; see \cite{koz77} and \cite{BirgetMMW94}, respectively.
We will give a short proof of these results based on the reduction from \cref{sub:pbm-hardness}.

\begin{theorem}[\cref{cor:main-intersection}]\label{thm:ia-hardness}
  The problem \dIEmpty{ia} is complete for \PSPACE{}.
  Moreover, this holds under the restriction that each automaton provided as part of the input has only two states, one of which is accepting.
\end{theorem}
\begin{proof}
  As in the proof of \cref{thm:pbm-hardness}, we show \PSPACE{}-hardness by polynomial-time many-one reduction from \dNCL{} to \dIEmpty{ia}.
  Let $(\Gamma, C_s, C_t)$ be an instance of the former.

  The idea is to encode a configuration $C \in \cC(\Gamma)$ into the states of a collection of two-state inverse automata.
  To this end, we introduce one such automaton $\mathcal{A}(c)$ for each local configuration $c \in \cC(\Gamma, v)$ at each vertex $v \in V(\Gamma)$.
  Its states $q_{\top}$ and $q_{\bot}$ indicate whether or not $C \vert_v = c$ holds.
  In particular, the starting state of $\cA(c)$ is chosen to be $q_{\top}$ if $C_s \vert_v = c$ and $q_{\bot}$ if $C_s \vert_v \neq c$.
  Similarly, its accepting state is $q_\top$ if $C_t \vert_v = c$ and $q_\bot$ if $C_t \vert_v \neq c$.

  As an alphabet $\Sigma$ for our automata, we use the collection of all $u(c_1, c'_2)$ as in the proof of \cref{thm:pbm-hardness} (with formation of inverses acting as involution).
  On input $u = u(c_1, c'_2)$, with $c_i, c'_i \in \cC(\Gamma, v_i)$ as in the proof of \cref{thm:pbm-hardness}, the transitions are defined as follows.
  The automata $\cA(c_i)$ with $i = 1,2$ transitions from $q_\top$ to $q_\bot$ via $u$, and the automata $\cA(c'_i)$ with $i = 1,2$ transitions from $q_\bot$ to $q_\top$ via $u$.
  For the automata $\cA(c)$ with $c \in \{ c_1,c_2, c'_1,c'_2\}$ these are the only defined transitions via $u$.
  If $c \in \cC(\Gamma, v_i)$ for some $i \in \{1,2\}$ and $c \not\in \{c_i, c'_i\}$, then $\cA(c)$ has a transition from $q_\bot$ to $q_\bot$ via $u$ and no other defined transitions via $u$.
  Finally, if $c \in \cC(\Gamma, v)$ with $v \not \in \{ v_1, v_2 \}$, then $\cA(c)$ has a transition from $q_\top$ to $q_\top$ and from $q_\bot$ to $q_\bot$ via $u$, i.e., $\cA(c)$ will remain in its current state upon reading $u$.

  It is now easy to see that a word $u_1u_2 \dotsc u_n \in \Sigma^\ast$ is contained in $\bigcap_c \cL(\cA(c))$, where the intersection extends over all local configurations $c \in \cC(\Gamma, v)$ at all vertices $v \in V(\Gamma)$, if and only if $u_1, u_2, \dotsc, u_n$ describes a sequence of transitions of $\Gamma$ from $C_s$ to $C_t$.
  As, moreover, the collection of automata $\cA(c)$ can be computed in polynomial time from $\Gamma$ and $C_s, C_t$, we conclude that \dIEmpty{ia} is indeed hard for \PSPACE.
  Since \dIEmpty{ia} is clearly contained in \PSPACE, the problem is \PSPACE-complete.
\end{proof}

\begin{corollary}[Bulatov, Kozik, Mayr, Steindl \cite{BulatovKMS16}]
  The problem \dIEmpty{dfa} is complete for \PSPACE{}.
  Moreover, this holds under the restriction that each automaton provided as part of the input has only three states, one of which is accepting.
\end{corollary}

Note that herein the transition functions of the automata are \emph{total} functions.
\begin{proof}
  We can convert each inverse automaton into a deterministic finite automaton with a total transition function by introducing a failure state and appropriate transitions to it.
\end{proof}

\subsection{The Subpower Membership Problem}\label{sub:subpower}

From \cref{thm:ia-hardness} we also derive a corresponding hardness result for the \emph{subpower membership problem} of an inverse semigroup $S$.
This problem is defined as follows.

\begin{decproblem}
	\iitem An integer $k$, a subset $\Sigma \sse S^k$, and an element $t \in S^k$.
	\qitem Is $t \in U$ where $U = \gen{\Sigma}$?
\end{decproblem}

Be aware that here $S$ is treated as a constant and \emph{not} part of the input.

\begin{corollary}[\cref{cor:main-subpower}]
  The subpower membership problem of an inverse semigroup~$S$ is complete for \PSPACE if and only if the Brandt monoid $B_2^1$ divides $S$; otherwise, it~is~in~\NC.
\end{corollary}
\begin{proof}
	The subpower membership problem clearly reduces to the problem $\dMemb[PB]{\vV}$ where $\vV$ is the variety generated by $S$. 
  Hence, the subpower membership problem is in \PSPACE and, if $S \in \vSI$ (\ie $B^1_2 \not\preccurlyeq S$ by \cref{lem:characterization-clifford}), then it is in \NC by \cref{cor:gb-complexity}.
	
	From \cref{thm:ia-hardness}, it follows easily that the subpower membership for $B_2^1$ is \PSPACE-complete (\eg using the same argument as \cite[Theorem 3.2]{BirgetM08} or \cite[Theorem 4.10]{Jack23}).
	Moreover, note that the target element $t$ obtained in the reduction does not project to the zero element of $B_2^1$ in any component of the direct product.
	Now, if $B_2^1$ divides $S$, then we can reduce the subpower membership problem of $B_2^1$ to the one of $S$ in the following way.
	
	As $B_2^1$ divides $S$, we can find some element $s \in S$ with $s\ov s \neq \ov s s$ and an idempotent $e \in E(S)$ with $e \geq s\ov s \vee \ov s s$.
	To see this, consider an inverse subsemigroup of $S$ projecting onto $B_2^1 = \gen{u, 1}$ and let $e$ be an arbitrary idempotent in the preimage $1$ and let $s = s' e$ where $s'$ is any preimage of $u$.
	Note that we can multiply the elements $\{e,s,\ov s, s\ov s, \ov s s\}$ as in $B_2^1$ as long as their product does not project to $0$ in $B_2^1$.
	Now, as the target element $t$ does not project onto the zero element of $B_2^1$ in any component, we can safely replace each occurrence of $1$ by $e$, of $u$ by $s$, and so on in every component of $t$.
  Performing the same substitution on all elements of $\Sigma$ 
  completes our reduction.
\end{proof}

\section{Further Related Problems}\label{sec:applications}

Finally, let us explore the consequences of our results to two further problems, namely the minimum generating set problem and the problem of solving equations.

\subsection{The Minimum Generating Set Problem}

As outlined in \cref{sec:related-work}, the minimum generating set has been first considered by Papadimitriou and Yannahakis \cite{PapadimitriouY96} and, in the Cayley table model, recently shown to be in \Ptime by Lucchini and Thakkar \cite{LucchiniT24} and even \NC \cite{CollinsGLW24}.
Yet, the complexity for arbitrary semigroups and also for permutation groups remains wide open.
In this section, we consider the minimum generating set problem for inverse semigroups in the partial bijection model. More formally, the minimum generating set problem (\dMGS{}) is defined as follows.

In this section, we consider the minimum generating set problem for inverse semigroups in the partial bijection model. More formally, the minimum generating set problem (\dMGS{}) is defined as follows.

\begin{decproblem}
	\iitem An inverse semigroup $U$ and an integer $k$.
	\qitem Is there some $\Xi \sse U$ with $\abs{\Xi} \leq k$ and $\gen{\Xi} = U$?
\end{decproblem}

Be aware that, while in the above we assumed without loss of generality that generating sets are closed under formation of inverses, for \dMGS{} we drop this assumption.
This is because adding inverses to $\Xi$, of course, changes $\abs{\Xi}$.
The hardness result below still applies to the variant of \dMGS{} where $\Xi$ is required to be inverse-closed; however, our reduction then no longer is in \ACz because we would need to count the number of self-inverse generators.

As for the membership problem, we write $\dMGS[PB]{}$ if $U$ is given by generators of an inverse subsemigroup of $\ISym(\Omega)$ for some $\Omega$ and denote the restriction to a variety $\vV$ by $\dMGS[PB]{\vV}$.

\begin{lemma}\label{lem:mgs-reduction}
	Let $\vV$ be a variety of finite inverse semigroups.
  Then the restricted membership problem $\dMembS[PB]{\vV}$ is \ACz-many-one-reducible to the problem $\dMGS[PB]{\vV \vee \vSl}$.
\end{lemma}

\begin{proof}
	Consider an instance $\Sigma \sse \ISym(\Omega)$ and $t \in \ISym(\Omega)$ for \dMembS[PB]{\vV}, \ie the question is whether $t \in U = \gen{\Sigma}$.
  We define a subsemigroup $U' \leq \ISym(\Omega')$ where $\Omega' = \Omega \cup (\Sigma \times \{1,2\})$. 
  To do so, we take $\Sigma' = \{t\} \cup (\Sigma \times \{1,2\})$ as a generating set, where $t$ is viewed as a partial bijection on $\Omega'$ by leaving it undefined outside of $\Omega$.
	For $(u,i) \in (\Sigma \times \{1,2\})$ and $x \in \Omega'$, we define $x^{(u,i)} = x^u$ if $x \in \Omega$ and $x^{(u,i)} = x$ if $x = (u,i)$ and otherwise $x^{(u,i)}$ is undefined.	
  Thus, in particular, $U' = \gen{\Sigma'}$ is isomorphic to an inverse subsemigroup of $\gen{ \{t\} \cup \Sigma} \times (Y_2 \times Y_2)^\Sigma$ and, as such, we have $U' \in \vV \vee \vSl$.
  Clearly, the set $\Sigma' \sse \ISym(\Omega')$ can be computed in $\ACz$.

	Let $k = \abs{\Sigma}$.
	We claim that $U' = \gen{\Sigma'}$ is generated by $2k$ elements if and only if $t \in U$. 
  To see this, first observe that the inverse subsemigroup $\gen{\Sigma \times \{1,2\}}$ of $U'$ projects onto the semilattice $E(\ISym(\Sigma \times \{1,2\}))$; hence, it cannot be generated by less than $2k$ elements. 
	Furthermore, we can find $U = \gen\Sigma$ as a subsemigroup of $U'$ as $u = (u,1) \ov{(u,2)}(u,2)$ for $u \in \Sigma$. 
	Thus, if $t \in U$, then $U'$ is generated by exactly $2k$ elements.
  On the other hand, observe that the canonical projection $\ISym(\Omega') \to \ISym(\Omega)$ maps $\gen{\Sigma \times \{1,2\}}$ onto $U$ and $t$ to $t$; hence, if $t \not \in U$, then $t \not \in \gen{\Sigma \times \{1,2\}}$.
  This completes the proof of the claim.
\end{proof}

\begin{remark}
	This reduction applies to arbitrary transformation semigroups.
	There is only one minor adjustment to be aware of: $U$ does not necessarily embed into $U'$; however, still all products of generators in $U$ of length at least two embed into $U'$.
  This is enough for the reduction to be correct if one checks upfront whether any of the generators coincide.
\end{remark}

\begin{corollary}[First Part of \cref{cor:main-mgs-equations}]
	Let $\vV$ be a variety of finite inverse semigroups. 
  Then the problem $\dMGS[PB]{\vV}$ is in \NP if $\vV\sse \vSI$ and \PSPACE-complete otherwise.
\end{corollary}
\begin{proof}
	It is clear that $\dMGS[PB]{\vV}$ is in $\NP^{\dMemb[PB]{\vV}}$ (just guess a generating set of an appropriate size and verify whether all of the original generators are in the inverse subsemigroup generated by the guessed generating set and vice versa). 
	If $\vV \not\sse \vSI$, then $\vV \not\sse \vG$ so $\vSl \sse \vV$.
  Hence, the \PSPACE-hardness of $\dMGS[PB]{\vV}$ follows from \cref{lem:mgs-reduction} and \cref{thm:pbm-hardness}.
\end{proof}

\subsection{Equations}

In this section we explore the consequences of our hardness result for the conjugacy problem to the related problem of deciding satisfiability of equations.
In particular, we are interested in the case, where the semigroup is part of the input.
We will see that in the partial bijection model, this variant of the problem is harder than deciding whether an equation has a solution in a fixed inverse semigroup.

Let $\cX$ be a set of variables.
An equation $\ell = r$ in an inverse semigroup $S$ is given as non-empty words $\ell, r \in (S \cup \cX \cup \ov \cX)^+$.
An assignment is a map $\sigma\colon \cX \to S$, which naturally extends to a homomorphism from $\sigma\colon (S \cup \cX \cup \ov \cX)^+ \to S$.
The problem $\dEqnSys{}$ of deciding whether a system of equations has a solution is defined as follows.

\begin{decproblem}
	\iitem An inverse semigroup $S$, and words $\ell_1, r_1, \ldots, \ell_k, r_k \in (S \cup \cX \cup \ov \cX)^+$.
	\qitem Is there a $\sigma\colon \cX \to S$ such that $\sigma(\ell_i) = \sigma(r_i)$ for all $1 \le i \le k$?
\end{decproblem}

We denote by \dEqnSys[CT]{} and \dEqnSys[PB]{} this problem in the Cayley table model and in the partial bijection model, respectively.
In the Cayley table model $S$ is given as a multiplication table.
In the partial bijection model the input is a set of generators $\Sigma \subseteq \ISym(\Omega)$ such that $\gen{\Sigma} = S$.
The problem of deciding whether a single equation has a solution occurs as a special case when $k = 1$.
We write \dEqn[CT]{} and \dEqn[PB]{}, respectively.

We will show below, that the $\PSPACE$-hardness of $\dEConjS[PB]{\vBM}$ can be transferred to \dEqn[PB]{}.
In addition, the problems of deciding whether a single equation or a system of equations have a solution are known to be $\NP$-hard for many fixed inverse semigroups.
We summarize these results here and give more details below.
If $\vV$ is a variety of finite inverse semigroups, then the following hold.
\begin{itemize}
  \item The problem $\dEqn[PB]{\vV}$ is $\PSPACE$-hard whenever $\vV \not\subseteq \vSI$.
  \item The problems $\dEqn[CT]{\vV}$ and $\dEqn[PB]{\vV}$ are $\NP$-hard whenever $\vV \not\subseteq \vGSol \vee \vBS$.
  \item The problems $\dEqnSys[CT]{\vV}$ and $\dEqnSys[PB]{\vV}$ are $\NP$-hard whenever $\vV \not\subseteq \vCom$.
\end{itemize}

Herein $\vGSol$ denotes the variety of finite solvable groups and $\vCom = \vAb \vee \vSl$ the variety of finite commutative inverse semigroups.
Note that these hardness results are matched by a corresponding upped bound, which is presented in \Cref{prop:eqn-alg} below.
The second and third statement are a consequence of known $\NP$-hardness results for fixed inverse semigroups.
We will briefly describe them before proving the first statement in \Cref{lem:eqn-hardness}.

Whenever $\vV \not\subseteq \vGSol \vee \vBS$, then either $\vV$ contains a non-solvable group or $B^1_2 \in \vV$.
According to Goldmann and Russell~\cite{GoldmannR02}, the problem of determining whether a single equation over a (fixed) finite group $G$ has a solution is $\NP$-complete for every non-solvable group $G$.
Likewise, the problem of determining whether a single equation over $B^1_2$ has a solution is also known to be $\NP$-complete \cite[Theorem 6]{BarringtonMMTT00}.
We thus conclude that both of the problems $\dEqn[CT]{\vV}$ and $\dEqn[PB]{\vV}$ are $\NP$-hard whenever $\vV \not\subseteq \vGSol \vee \vBS$.

The third statement is a special case of a dichotomy for regular semigroups: deciding whether a restricted system of equations, where the right-hand side is a constant, has a solution is in polynomial time if $S$ is in the variety of finite semigroups generated by abelian groups and regular bands, and $\NP$-complete otherwise \cite{KlimaTT07}.
The statement follows by restriction to inverse semigroups.

\begin{lemma}\label{lem:eqn-hardness}
  Let $\vV$ is a variety of finite inverse semigroups.
  Then the problem $\dEqn[PB]{\vV}$ is hard for $\PSPACE$ whenever $\vV \not\subseteq \vSI$.
\end{lemma}

\begin{proof}
  We reduce from $\dEConjS[PB]{\vV}$, which is $\PSPACE$-hard by~\cref{thm:pbm-hardness}.
  We are given two idempotents $e_s, e_t$ that are conjugate (hence, $\gJ$-equivalent) in $\ISym(\Omega)$.
  By definition $e_s \sim_U e_t$ if and only if there is some $X \in U$ with $\bar X e_s X = e_t$ and $X e_t \bar X = e_s$.
  As $e_s$ and $e_t$ are idempotent, this holds if and only if there is some $X\in  \gen{U \cup \{e_s,e_t\}} \in \vV$ with $\bar X e_s X = e_t$ and $X e_t \bar X = e_s$.
  Moreover, by \cref{lem:green-semi_rel} and \cref{lem:idem_conj_is_green_j}, either one of the two equations has a solution if and only if the other one does.
	Hence, $\dEqn[PB]{\vV}$ is hard for $\PSPACE$.
\end{proof} 

\begin{observation}\label{prop:eqn-alg}
  Let $\vV$ be a variety of finite inverse semigroups.
	The problem $\dEqnSys[CT]{\vV}$ is in $\NP$ and $\dEqnSys[PB]{\vV}$ is in $\NP^{\dMemb[PB]{\vV}}$, \ie solvable in $\NP$ with an oracle for $\dMemb[PB]{\vV}$.

  \smallskip
	
	Moreover, $\dEqnSys[PB]{\vV}$ remains in $\NP^{\dMemb[PB]{\vV}}$ if each variable can be constrained to some inverse subsemigroup and we allow arbitrary constants from $\ISym(\Omega)$ (not restricted to $\vV$).
\end{observation}

In the partial bijection model we obtain that $\dEqnSys[PB]{}$ is in $\PSPACE$ with \cref{thm:pbm-hardness}, and that $\dEqnSys[PB]{\vSI}$ is in $\NP$ with \cref{thm:gb-complexity}.

\begin{proof}
	The algorithm follows the guess-and-check pattern: we guess an assignment $\sigma\colon \cX \to S$ and then verify that it satisfies the equations.
	In the Cayley table model this is straight-forward and so it only remains to consider the partial bijection model.
	There, we guess a map $\sigma\colon \cX \to \ISym(\Omega)$ and check whether $\sigma(X) \in \gen{\Sigma}$ for each $X \in \cX$ using the $\dMemb[PB]{\vV}$ oracle.
	Then we verify that $\sigma(\ell_i) = \sigma(r_i)$ for all $1 \le i \le k$.
  The evaluations $\sigma(\ell_i)$ and $\sigma(r_i)$ can clearly be computed in polynomial time (in either input model).

	Allowing inverse subsemigroup constraints does not increase the complexity since we already use the membership oracle anyway.
	Allowing arbitrary constants from $\ISym(\Omega)$ also does not increase difficulty, as evaluations are computed in $\ISym(\Omega)$ either way.
\end{proof}

Combining \cref{lem:eqn-hardness} with \cref{prop:eqn-alg}, we obtain the following corollary, which constitutes the second part of \cref{cor:main-mgs-equations}.

 \begin{corollary}\label{cor:equations}
 	Let $\vV$ be a variety of finite inverse semigroups.
 	Then the problems $\dEqn[PB]{\vV}$ and $\dEqnSys[PB]{\vV}$ are in \NP if $\vV\sse \vSI$ and \PSPACE-complete otherwise.
 \end{corollary}

\section{Discussion and Open Problems}\label{sec:discussion}

By investigating the membership problem in inverse semigroups, we filled a gap between the rather restricted case of groups and the very general case of arbitrary semigroups.
We gave a classification of the complexity of membership and conjugacy in inverse semigroups according to their combinatorial structure.
Here the combinatorial Brandt semigroup and monoid are the critical obstructions for the membership problem being easy.
Furthermore, by applying these results, we gained new insights on the complexity of closely related problems such as the intersection non-emptiness problem for inverse automata, the minimum generating set problem, and the equation satisfiability problem.

Applying \cref{thm:main-CT} and \cref{thm:main-PB} to the case of aperiodic inverse semigroups shows that the membership and conjugacy problems are either in \ACz, or \LOGSPACE-complete, or \PSPACE-complete (with \PSPACE-complete only in the partial bijection model).
Thus, in comparison to the classification for varieties of aperiodic monoids \cite{BeaudryMT92}, we additionally have the \LOGSPACE-complete case, whereas we do not have the \Ptime-complete, \NP-complete, and \NP-hard cases.

A corresponding classification of the complexity of the membership problem for varieties of arbitrary semigroups is still a major endeavor.
While there is the classification for aperiodic monoids by Beaudry, McKenzie, and Thérien~\cite{BeaudryMT92} mentioned above, the case for semigroups is considerably more involved as there are many more varieties of finite semigroups than of finite monoids.
Moreover, it remains to integrate the case of groups into the consideration.
\begin{problem}
	Give a classification of the varieties of finite semigroups in terms of their complexity for the membership problem.
\end{problem}

Another class of structures, situated between inverse semigroups and arbitrary semigroups, that we would like to draw attention to is the class of regular $\ast$-semigroups introduced by Nordahl and Scheiblich~\cite{NordahlScheiblich78}.
These are regular semigroups with distinguished inverses $x \mapsto x^\ast$ such that $x^{\ast\ast} = x$, $(xy)^\ast = y^\ast x^\ast$, and $x x^\ast x = x$.
The difference to inverse semigroups is that inverses need not be unique or, equivalently, that idempotents need not commute. 
As such, regular $\ast$-semigroups admit a far richer combinatorial structure.

\begin{problem}
	Give a classification of the varieties of finite regular $\ast$-semigroups in terms of their complexity for the membership problem.
\end{problem}

Our algorithms for the \NC (resp.\ \NP) cases of our dichotomy result for the partial bijection model are very efficient in the sense that they provide reductions computable in \LOGSPACE and also in linear or quadratic time to \dUGAP as well as to the respective problems for groups.
Thus, the only open questions about the complexity of these problems come from the group case.
For the membership problem this leads to the following rather far-reaching question.

\begin{question}
	Is \dMemb[PB]{\vG} in \LOGSPACE?
\end{question}

We do not feel confident to make any guess about this question and want to use the present work to foster further research on this topic.
On the other hand, we believe that the answer to the following question concerning the Cayley table model is likely to be negative.

\begin{question}
	Is \dMemb[CT]{\vG} in \ACz?
\end{question}

Another question is whether the $\bigO(\log^2 n)$ bound due to Babai and Szemer\'edi~\cite{BabaiS84} on the length of straight-line programs in groups of order $n$  is asymptotically optimal.
In some special cases, like for Abelian groups, an (asymptotically optimal) $\bigO (\log n)$ bound can be obtained instead.
Thus, the question here is whether this can be extended to \emph{all} groups.

For inverse semigroups, we obtained a complete characterization of varieties admitting polylogarithmic SLPs in \cref{cor:short-SLP-char}.
A similar result for monoids was obtained by \ifAnonimous Fleischer\else the first author\fi~\cite{Fleischer19diss}.
Therefore, the natural question is to ask for a similar characterization for all semigroups. 
For some preliminary results in that direction, see also \cite{Fleischer19diss}.

\bibliography{bibexport}

\end{document}